\documentclass[12pt,english]{article}
\usepackage{amsmath}
\usepackage{amssymb}
\usepackage{color}
\usepackage[colorinlistoftodos]{todonotes}
\usepackage{color,amsmath}
\usepackage{pgf,tikz}
\usetikzlibrary{arrows}
\usetikzlibrary{decorations} 
\usetikzlibrary{shapes,snakes,arrows}
\usepackage{multicol}
\usepackage{subfigure}

\usetikzlibrary{automata}
\usetikzlibrary{positioning}
\usetikzlibrary{shapes.multipart}
\usepackage{rotating}

\usepackage{booktabs}
\usepackage{multirow}
\usepackage{lscape}

\usepackage{mdwlist}
\usepackage{footmisc}
\usepackage{amsthm}
\usepackage{ulem}

\newcommand*{\field}[1]{\mathbb{#1}}

\newtheorem{lemma}{Lemma}
\newtheorem{theorem}{Theorem}
\newtheorem{corollary}{Corollary}
\newtheorem{proposition}{Proposition}
\newtheorem{remark}{Remark}\rm
\newtheorem{example}{Example}\rm

\newcommand{\setSC}{\mathcal{SC}}
\newcommand{\tR}{\tilde{R}}
\usepackage[ruled,vlined]{algorithm2e}

\usepackage{floatrow}
\newfloatcommand{capbtabbox}{table}[][\FBwidth]
\usepackage{blindtext}
\usepackage{bm}

\renewcommand{\emph}[1]{{\textit{#1}}}

\usepackage{float}

\begin{document}

\title{Shapley-Scarf Housing Markets: Respecting Improvement,  Integer Programming, and Kidney Exchange\footnote{This work is financed by 
COST Action CA15210 ENCKEP, supported by COST (European Cooperation in Science and Technology) -- http://www.cost.eu/.}}

\author{P\'eter Bir\'o\thanks{Centre for Economic and Regional Studies, and Corvinus University of Budapest, 1097 Budapest, Tóth Kálmán utca 17, Hungary. P.\ Bir\'o is supported by the Hungarian Scientific Research Fund -- OTKA (no.\ K129086).} \and Flip Klijn\thanks{Institute for Economic Analysis (CSIC) and Barcelona GSE, Campus UAB,
08193 Bellaterra (Barcelona), Spain; e-mail: \texttt{flip.klijn@iae.csic.es}. F.\ Klijn gratefully acknowledges financial support from AGAUR--Generalitat de Catalunya (2017-SGR-1359) and the Spanish Ministry of Science and Innovation through grant ECO2017-88130-P AEI/FEDER, UE and the Severo Ochoa Programme for Centres of Excellence in R\&D (CEX2019-000915-S).} \and  Xenia Klimentova\thanks{INESC TEC, Porto, Portugal.} 
\and Ana Viana\thanks{INESC TEC and ISEP -- School of Engineering, Polytechnic of Porto, Porto, Portugal.}}

\date{\today}

\maketitle    
\begin{abstract}
\noindent In a housing market of Shapley and Scarf \cite{SS1974}, each agent is endowed with one indivisible object and has preferences over all objects. An allocation of the objects is in the (strong) core if there exists no (weakly) blocking coalition. In this paper we show that in the case of strict preferences the unique strong core allocation (or competitive allocation) ``respects improvement'': if an agent's object becomes more attractive for some other agents, then the agent's allotment in the unique strong core allocation weakly improves. We obtain a general result in case of ties in the preferences and provide new integer programming formulations for computing (strong) core and competitive allocations. Finally, we conduct computer simulations to compare the game-theoretical solutions with maximum size and maximum weight exchanges for markets that resemble the pools of kidney exchange programmes.
\vspace*{0.1cm}

\end{abstract}
\renewcommand{\thefootnote}{\arabic{footnote}}

\section{Introduction}\label{sec:intro}
\vspace{-3mm}
\par

Shapley and Scarf \cite{SS1974} introduced so-called ``housing markets'' to model trading in commodities that are inherently indivisible. Specifically, in a housing market each agent is endowed with an object (e.g., a house or a kidney donor) and has ordinal preferences over all objects, including her own. The aim is to find plausible or desirable allocations where each agent is assigned one object. A standard approach in the literature is to discard allocations that can be blocked by a coalition of agents. Specifically, a coalition of agents blocks an allocation if they can trade their endowments so that each of the agents in the coalition obtains a strictly preferred allotment. Similarly, a coalition of agents weakly blocks an allocation if they can trade their endowments so that each of the agents in the coalition obtains a weakly preferred allotment and at least one of them obtains a strictly preferred allotment. Thus, an allocation is in the (strong) core if it is not (weakly) blocked. A distinct but also well-studied solution concept is obtained from competitive equilibria, each of which consists of a vector of prices for the objects and a (competitive) allocation such that each agent's allotment is one of her most preferred objects among those that she can afford.
Interestingly, the three solution concepts are entwined: the strong core is contained in the set of competitive allocations, and each competitive allocation pertains to the core.

In a separate line of research, Balinski and S\"onmez \cite{BS1999} studied the classical two-sided college admissions model of Gale and Shapley \cite{GaleShapley1962} and proved that the student-optimal stable matching mechanism (SOSM) \emph{respects improvement} of student's quality. This means that under SOSM, an improvement of a student's rank at a college will, ceteris paribus, lead to a weakly preferred match for the student. 
The natural transposition of this property to (one-sided) housing markets requires that an agent obtains a weakly preferred allotment whenever her object becomes more desirable for other agents.
We study the following question: Do the most prominent solution concepts for Shapley and Scarf's \cite{SS1974} housing market ``respect improvement''? We obtain several positive answers to this question, which we describe in more detail in the next subsection.

The respecting improvement property is important in many applications where centralised clearinghouses use mechanisms to implement barter exchanges. A leading example are kidney exchange programmes (KEPs), where end-stage renal patients exchange their willing but immunologically incompatible donors (Roth et al.\ \cite{RSU2004}).
In the context of KEPs, the respecting improvement property means that whenever a patient brings a ``better'' donor (e.g., younger or with universal blood type 0 instead of A, B, or AB) or registers an additional donor, the KEP should assign her the same or a better exchange donor.
However, in current KEPs, the typical objective is to maximise the number of transplants and their overall qualities (see, e.g., \cite{biroetal2019b}) which can lead to violations of the respecting improvement property. As an illustration, consider the maximisation of the number of transplants in Figure~\ref{fig:exampleintro}, where each node represents a patient-donor pair. Directed edges represent compatibility between the donor in one pair and the patient in another, and patients may have different levels of preference over their set of compatible donors. %
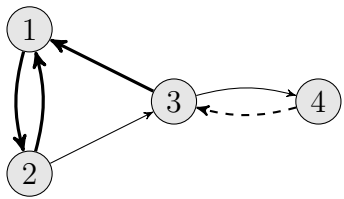
\begin{figure}[h]
    \centering
    \begin{tikzpicture}[scale=1.2,->, >=stealth',scale = 0.2, inner sep=2mm]
\tikzstyle{nd}=[circle,draw,fill=black!10,inner sep=0pt,minimum size=6mm]
 \node[nd] (1) at (0,4)  {$1$};
 \node[nd] (2) at (0,-4) {$2$};
 \node[nd] (3) at (8,0) {$3$};
 \node[nd] (4) at (16,0) {$4$};
 \path[every node/.style={font=\sffamily\footnotesize}]
    (1) edge[bend right=15,very thick] (2) 
    (2) edge[bend right=15,very thick] (1) 
    (3) edge [very thick] (1) 
    (2) edge (3) 
    (3) edge[bend left = 15]  (4) 
    (4) edge[dashed, bend left = 15, thick] (3) 
    ;
 \end{tikzpicture} 
 
    \caption{The maximisation of the number of transplants does not respect improvement.}   
    \label{fig:exampleintro}
\end{figure}
Initially there are only continuous edges, where a thick (thin) edge points to the most (least) preferred donor. For example, patient 3 has two compatible donors: donors 1 and 4, and donor 1 is preferred to donor 4. Obviously, the unique way to maximise the number of (compatible) transplants is obtained by picking the three-cycle (1,2,3). Suppose that patient 4 receives antigen desensitisation treatment so that donor 3 becomes compatible for her, or patient 3 succeeds in bringing a second donor to the KEP and this donor turns out to be compatible for patient 4. Then, the discontinuous edge is included and patient-donor pair 3 ``improves.'' But now the unique way to maximise the number of (compatible) transplants is obtained by picking the 2 two-cycles (1,2) and (3,4), which means that patient 3 receives a kidney that is strictly worse than the kidney she would have received initially.

Similarly, the allocations induced by the standard objectives of KEPs need not be in the (strong) core. We refer to Example~\ref{example:ill} for an illustration of this for the case of the maximisation of the number of transplants. As a consequence, blocking coalitions may exist.
This is an undesirable feature because patient groups could make a potentially justified claim that the  matching procedure is not in their best interest. A particular instance could occur in the organisation of international kidney exchanges if a group of patient-donor pairs, all citizens of the same country, learn that an internal (i.e., national) matching would all of them give a more preferred kidney.

We conduct simulations to determine to what extent the typical KEP objectives lead to the violation of the respecting improvement property and the existence of blocking coalitions.
The simulations also include the three standard game-theoretical solution concepts (as they also do not ``completely'' satisfy the respecting improvement property), for which we first develop novel integer programming formulations to speed up the computations. 

Next, we describe our contributions in more detail and review the related literature.

\subsection{Contributions}

Section~\ref{sec:RI} contains our theoretical results on the respecting improvement property. First, we show that for strict preferences the unique strong core allocation (which coincides with the unique competitive allocation) respects improvement (Theorem~\ref{theorem:strictRI}). 
In the case of preferences with ties, since the strong core can be empty, we focus on the set of competitive allocations. Since typically multiple competitive allocations exist, we have to make setwise comparisons.
We obtain a natural extension of our first result: focusing on the agent's allotments obtained at competitive allocations, we establish that her most preferred allotment in the new market is weakly preferred to her most preferred allotment in the initial market; and similarly, her least preferred allotment in the new market is weakly preferred to her least preferred allotment in the initial market (Proposition~\ref{proposition:compRI}).
Finally, we also prove that when preferences have ties the strong core respects improvement conditional on the strong core being non-empty. More precisely, under the assumption that strong core allocations exist in both the initial and new markets, we show that the agent under consideration weakly prefers each allotment in the new strong core to each allotment in the initial strong core (Theorem~\ref{theorem:weakRI}).

Then we tackle an important assumption in the housing market of Shapley and Scarf, namely that allocations can contain exchange cycles of any length, i.e., cycles are unbounded. As a consequence, some allocations obtained from the theory of housing markets might be difficult to implement in the case of KEPs. The reason is that all transplants in a cycle are usually carried out simultaneously to avoid reneging. So, if the number of surgical teams and operation rooms is small, some of the transplants have to be conducted necessarily in a non-simultaneous way. In many countries, this ``risky'' solution is not allowed \cite{biroetal2019a}.
The definition of core, set of competitive allocations, and strong core can be adjusted to the requirement that the length of exchange cycles does not exceed an exogenously given maximum.
However, in this case the core (and hence also the set of competitive allocations and the strong core) can be empty.\footnote{The corresponding decision problem is NP-hard \cite{BMcD2010,Huang2010} even for tripartite graphs (also known as the cyclic 3D stable matching problem \cite{NG1991}).} Conditional on the existence of a core, competitive, or strong core allocation, we show that even if preferences are strict, when the length of exchange cycles is limited (upper bound 3 or higher), the core, the set of competitive allocations, and the strong core do not respect improvement in terms of the most preferred allotment (Proposition~\ref{proposition:length3}).

In view of practical applications such as KEPs, we provide as a second main contribution, in Section~\ref{sec:IP}, novel integer programming (IP) formulations for finding core, competitive, and strong core allocations. Our simple sets of constraints for the three solution concepts clearly show the hierarchy between them by pinpointing the additional requirements needed when moving from one solution concept to a stronger one. Our formulations are concise and useful for practical computations. 


%

Section~\ref{sec:comp} contains our third main contribution, which complements our theoretical analysis and consists of computer simulations comparing core/competitive/strong core with maximum number of transplants and maximum total weight allocations for both bounded length and unbounded length exchanges. The total weight of an exchange is the sum of the weights associated to the arcs involved in the exchange. In the simulations we use the IP models developed in \cite{Klimentova2020stable} for bounded length exchanges. For unbounded length exchanges we use our novel IP formulations. To carry out our simulations we draw markets from pools similar to those observed in KEPs. 
We analyse the impact in the objective functions of the stability requirements associated with core, competitive, and strong core allocations. For unbounded length we also study the \emph{price of fairness}: the decrease in the number of transplants of maximum weight, core, competitive and strong core allocations, when compared with the maximum size solution. The analysis proceeds with an indirect assessment of how far other solutions would be from the strong core. Such indicator provides some insight  in/into the number of patients in a pool that could get a strictly better match. The section  concludes with a computational analysis of the frequency of violations of the respecting improvement property. It was observed that even though there exist cases of violation of the property for (Wako-, strong) core, their quantity is dramatically lower that those for the maximum size and maximum size solutions.

\subsection{Literature review}

\paragraph{Housing markets}

The non-emptiness of the core was proved in \cite{SS1974} by showing the balancedness of the corresponding NTU-game, and also in a constructive way, by showing that David Gale's  famous Top Trading Cycles algorithm (TTC) always yields competitive allocations. Roth and Postlewaite \cite{RP1977} later showed that for strict preferences the TTC results in the unique strong core allocation, which coincides with the unique competitive allocation in this case. However, if preferences are not strict (i.e., ties are present), the strong core can be empty or contain more than one allocation, but the TTC still produces all competitive allocations.
Wako \cite{Wako1984} showed that the strong core is always a subset of the set of competitive allocations.
Quint and Wako \cite{QW2004} provided an efficient algorithm for finding a strong core allocation whenever there exists one. Their work was further generalised and simplified by Cechl\'arov\'a and Fleiner \cite{Cechlarova2010} who used graph models.
Wako~\cite{Wako1999} showed that the set of competitive allocations coincides with a core defined by antisymmetric weak domination. This equivalence is key for our extension of the definition of competitive allocations to the case of bounded exchange cycles.

\paragraph{Respecting improvement} 

For Gale and Shapley's \cite{GaleShapley1962}  college admissions model, Balinski and S\"onmez \cite{BS1999} proved that the student-optimal stable matching mechanism (SOSM) respects improvement of student's quality. 
Kominers \cite{Kominers2019} generalised this result to more general settings. Balinski and So\"nmez \cite{BS1999} also showed that SOSM is the unique stable mechanism that respects improvement of student quality. Abdulkadiroglu and So\"nmez \cite{AS2003} proposed and discussed the use of TTC in a model of school choice, which is closely related to the college admissions model. Klijn \cite{Klijn2019} proved that TTC respects improvement of student quality.

Focusing on the other side of the market, Hatfield et al. \cite{HKN2016} studied the existence of mechanisms that respect improvement of a \emph{college's} quality. The fact that colleges can match with multiple students leads to a strong impossibility result:  Hatfield et al. \cite{HKN2016} proved that there is no stable nor Pareto-efficient mechanism that respects improvement of a college's quality. In particular, the (Pareto-efficient) TTC mechanism does not respect improvement of a college's quality.

In the context of KEPs with pairwise exchanges, the incentives for bringing an additional donor to the exchange pool was first studied by Roth et al.\ \cite{RSU2005}. In the model of housing markets their \emph{donor-monotonicity} property boils down to the respecting improvement property. Roth et al.\ \cite{RSU2005} showed that so-called priority mechanisms are donor-monotonic if each agent's preferences are dichotomous, i.e., she is indifferent between all acceptable donors.
However, if agents have non-dichotomous preferences, then any mechanism that maximises the number of pairwise exchanges (so, in particular any priority mechanism) does not respect improvement. This can be easily seen by means of Example~\ref{example:pairwise1} in Section~\ref{sec:RI}.

\paragraph{IP formulations for matching}

Quint and Wako \cite{QW2004} already gave IP formulations for finding core and strong core allocations, but the number of constraints in their paper is highly exponential, as their formulations contain a no-blocking condition for each set of agents and any possible exchanges among these agents. Other studies  provided IP formulations for other matching problems.
In particular, for Gale and Shapley's \cite{GaleShapley1962} college admissions model, Ba\"iou and Balinski \cite{baiou2000stable} already described the stable admissions polytope, which can be used as a basic IP formulation. Further recent papers in this line of research focused on college admissions with special features \cite{ABMcB2016}, stable project allocation under distributional constraints \cite{ABSz2018}, the hospital--resident problem with couples \cite{BMMcB2014}, and ties \cite{KM2014,Delormeetal2019}.

\paragraph{Kidney exchange programmes}

Starting from the seminal work by Saidman et al.\ \cite{Saidman06}, initial research on KEPs focused on integer programming (IP) models for selecting pairs for transplantation in such a way that maximum (social) welfare, generally measured by the number of patients transplanted, is achieved. Authors in \cite{Constantino2013, DickersonMPST16, makhau15} proposed new, compact formulations that, besides extending the models in \cite{Saidman06} to accommodate for non-directed donors and patients with multiple donors, also aimed to efficiently solve problems of larger size.  Later studies~\cite{JPP16,dickerson_failure-aware_2014,Carvalhoetal2020,McElfresh2019}
modelled the possibility of pair dropout or cancellation of transplants (if e.g.\ new incompatibilities are revealed). While~\cite{JPP16} and \cite{dickerson_failure-aware_2014} aimed to find a solution that maximises expected welfare, \cite{Carvalhoetal2020} proposed robust optimisation models that search for a solution that, in the event of a cancellation, can be replaced by an alternative (recourse) solution that in terms of selected patients is as ``close'' as possible to the initial solution. McElfresh et al.\ \cite{McElfresh2019} also addressed the robustness of solutions, but they did not consider the possibility of recourse.

In a different line of research, \cite{Carv17, fair19, Mincu2020} studied KEPs where agents (e.g.\ hospitals, regional and national programmes) can collaborate. Allowing agents to control their internal exchanges, Carvalho et al.\ \cite{Carv17} studied strategic interaction using non-cooperative game theory. Specifically, for the two-agent case, they designed a game such that some Nash equilibrium maximises the overall social welfare. Considering multiple matching periods, Klimentova et al.\ \cite{fair19} assumed agents to be non-strategic. Taking into account that at each period there can be multiple optimal solutions, each of which can benefit different agents, the authors proposed an integer programming model to achieve an overall fair allocation. 
Finally, Mincu et al.\ \cite{Mincu2020} proposed integer programming formulations for the case where optimisation goals and constraints can be distinct for different agents.

A recent line of research acknowledges the importance of considering patients' preferences (associated with e.g.\ graft quality) over matches. Bir\'o and Cechl\'arov\'a \cite{BC2007} considered a model for unbounded length kidney exchanges, where patients most care about the graft they receive, but as a secondary factor (whenever there is a tie) they prefer to be involved in an exchange cycle that is as short as possible. The authors showed that although core allocations can still be found by the TTC algorithm, finding a core allocation with maximum number of transplants is a computationally hard problem (inapproximable, unless $P=NP$).
Recently, Klimentova et al.\ \cite{Klimentova2020stable} provided integer programming formulations when each patient has preferences over the organs that she can receive. The authors focused on allocations that among all (strong) core allocations have 
maximum cardinality. Moving away from the (strong) core, they also analysed the trade-off between maximum cardinality and the number of blocking cycles. Finally, the reader is referred to \cite{biroetal2019b} for a recent review of KEPs.

\section{Preliminaries}\label{sec:defs}

We consider housing markets as introduced by Shapley and Scarf \cite{SS1974}. Let $N=\{1,\ldots,n\}$, $n\geq 2$, be the set of \emph{agents}. Each agent $i\in N$ is endowed with one object denoted by $e_i=i$. Thus, $N$ also denotes the set of \emph{objects}. Each agent $i\in N$ has complete and transitive \emph{(weak) preferences} $R_{i}$ over objects. We denote the strict part of $R_i$ by $P_{i}$, i.e., for all $j,k\in N$, $j P_i k$ if and only if $j R_i k$ and not $k R_i j$. Similarly, we denote the indifference part of $R_i$ by $I_{i}$, i.e., for all $j,k\in N$, $j I_i k$ if and only if $j R_i k$ and $k R_i j$. Let $R\equiv {(R_i)}_{i\in N}$. A \emph{(housing) market} is a pair $(N,R)$. Object $j\in N$ is \emph{acceptable} to agent $i\in N$ if $j R_i i$. Agent $i$'s preferences are called \emph{strict} if they do not exhibit \emph{ties} between acceptable objects, i.e., for all acceptable $j,k\in N$ with $j\neq k$ we have $j P_i k$ or $k P_i j$. A housing market has strict preferences if each agent has strict preferences. A housing market where agents do not necessarily have strict preferences is referred to as a housing market with weak preferences.

Given a housing market $M=(N,R)$ and a set $S\subseteq N$, the \emph{submarket} $M_S$ is the housing market where $S$ is the set of agents/objects and where the preferences ${(R_i)}_{i\in S}$ are restricted to the objects in $S$. 

The \emph{acceptability graph} of a housing market $M=(N,R)$ is the directed graph 
$G_M=(N,E)$, 
or $G$ for short, where the set of nodes is $N$ and where $(i,j)$ is a directed edge in $E$ if $j$ is an acceptable object for $i$, i.e., $j R_i i$. In particular, all self-loops $(i,i)$ are in the graph (but for convenience they are omitted in all figures). Let $\tilde{N}\subseteq N$ and $\tilde{E}\subseteq E \cap (\tilde{N} \times \tilde{N})$. For each $i\in \tilde{N}$, the set of agent $i$'s \emph{most preferred edges} in graph $\tilde{G}\equiv (\tilde{N},\tilde{E})$ or simply $\tilde{E}$ is the set $\tilde{E}^{T,i} \equiv \{(i,j) : (i,j)\in \tilde{E} \mbox{ and for each } (i,k)\in \tilde{E}, j R_i k\}$. The most preferred edges in graph $\tilde{G}$ is the set $\cup_{i\in \tilde{N}} \tilde{E}^{T,i}$.

Let $M=(N,R)$ be a housing market. An allocation is a redistribution of the objects such that each agent receives exactly one object. Formally, an \emph{allocation} is a vector $x=\left(x_i\right)_{i\in N}\in N^{N}$ such that:
\begin{itemize*}
\item[(i)] for each $i\in N$, $x_i\in N$ denotes agent $i$'s \emph{allotment}, i.e., the object that she receives, and
\item[(ii)] no object is assigned to more than one agent, i.e.,
$\cup_{i\in N}\{x_i\}=N.$
\end{itemize*}
We will focus on \emph{individually rational} allocations, i.e., allocations where each agent receives an acceptable object. Then, an allocation $x$ can equivalently be described by its corresponding \emph{cycle cover} $g^x$ of the acceptability graph $G$. Formally, $g^x=(N,E^x)$ is the subgraph of $G$ where $(i,j)\in E^x$ if and only if $x_i=j$. Thus, the graph $g^x$ consists of disconnected \textit{trading cycles} or \textit{exchange cycles} that cover $G$. We will often write an (individually rational) allocation in cycle-notation, i.e., as a set of exchange cycles (where we sometimes omit self-cycles). We refer to Example~\ref{example:ill} for an illustration.

An allocation $x$ Pareto-dominates an allocation $z$ if for each $i\in N$, $x_i R_i z_i$, and for some $j\in N$, $x_j P_j z_j$. An allocation is \emph{Pareto-efficient} if it is not Pareto-dominated by any allocation. Two allocations $x,z$ are \emph{welfare-equivalent} if for each $i\in N$, $x_i I_i z_i$.


Next, we recall the definition of solution concepts that have been studied in the literature. A non-empty coalition $S\subseteq N$ \emph{blocks} an allocation $x$ if there is an allocation $z$ such that
(1) $\{z_i : i\in S\} = S$ and
(2) for each $i\in S$, $z_i P_i x_i$.
An allocation $x$ is in the \emph{core}\footnote{In the literature the core is sometimes called the weak core or ``regular'' core.} of the market if there is no coalition that blocks $x$.

A non-empty coalition $S\subseteq N$ \emph{weakly blocks} an allocation $x$ if there is an allocation $z$ such that
(1) $\{z_i : i\in S\} = S$, (2) for each $i\in S$, $z_i R_i x_i$, and (3) for some $j\in S$, $z_j P_j x_j$.
An allocation $x$ is in the \emph{strong core}\footnote{In the literature the strong core is sometimes called the strict core.} of the market if there is no coalition that weakly blocks $x$.

A price-vector is a vector $p={(p_i)}_{i\in N}\in \field{R}^N$ where $p_i$ denotes the price of object $i$.
A competitive equilibrium is a pair $(x,p)$ where $x$ is an allocation and $p$ is a price-vector such that:
\begin{itemize*}
\item[(i)] for each agent $i\in N$, object $x_i$ is affordable, i.e., $p_{x_i} \leq p_i$ and
\item[(ii)] for each agent $i\in N$, each object she prefers to $x_i$ is not affordable, i.e., $j P_i x_i$ implies 
$p_j> p_i$.
\end{itemize*}
An allocation is a \emph{competitive allocation} if it is part of some competitive equilibrium. Since there are $n$ objects, we can assume, without loss of generality, that prices are integers in the set $\{1,2,\ldots,n\}$.
\begin{remark}\label{remark:prices}
{\rm If $(x,p)$ is such that
\vspace*{-0.2cm}
\begin{itemize*}
\item for each $i\in N$, $p_{x_i} \leq p_{i}$, or
\item for each $i\in N$, $p_{i} \leq p_{x_i}$,
\end{itemize*}
\vspace*{-0.2cm}
then for each $i\in N$, $p_{x_i} = p_i$. This follows immediately by looking at each exchange cycle separately (see, e.g., the proof of Lemma~1 in \cite{Cechlarova2010}). Hence, at each competitive equilibrium $(x,p)$, for each $i\in N$, $p_{x_i} = p_i$.
}
\end{remark}
%

Wako \cite{Wako1999} proved that the set of competitive allocations can be defined equivalently as a different type of core.
Formally, a non-empty coalition $S\subseteq N$ \emph{antisymmetrically weakly blocks} an allocation $x$ if there is an allocation $z$ such that
(1) $\{z_i : i\in S\} = S$, (2) for each $i\in S$, $z_i R_i x_i$, (3) for some $j\in S$, $z_j P_j x_j$, and (4) for each $i\in S$, if $z_iI_ix_i$ then $z_i=x_i$.
Requirements (1--3) say that coalition $S$ weakly blocks $x$. The additional requirement (4) is that if an agent in $S$ is indifferent between her allotments at $x$ and $z$ then she must get the very same object, i.e., $z_i=x_i$. An allocation $x$ is in the \emph{core defined by antisymmetric weak domination} if there is no coalition that antisymmetrically weakly blocks $x$. Wako \cite{Wako1999} proved that the set of competitive allocations coincides with the core defined by antisymmetric weak domination. Henceforth, we will often refer to the set of competitive allocations as the \emph{Wako-core}.

\begin{lemma}\label{lemma:cycles}
The strong core, the set of competitive allocations (i.e., the Wako-core), and the core consist of individually rational allocations. Moreover, the cores are equivalently characterised by the absence of \emph{blocking cycles} in the acceptability graph $G=(N,E)$. In other words, in the definition of each of the three cores, it is sufficient to require no-blocking by coalitions $S$, say $S=\{i_1,\ldots,i_k\}$, such that for each $l=1,\ldots,k$ (mod $k$), $z_{i_l}=i_{l+1}$ and $(i_l,i_{l+1})\in E$.
\end{lemma}
\begin{proof}
Individual rationality is immediate.
To prove the statement for the strong core, let $x$ be an individually rational allocation. Suppose there is a non-empty coalition $T$ that weakly blocks $x$ through some allocation $w$. Let $j\in T$ be such that $w_j P_j x_j$. Let $S\subseteq T$ be the agents that constitute the exchange cycle, say $(i_1,\ldots,i_k)$, in $w$ that involves agent $j$, i.e., without loss of generality, $j=i_1$. Since $w$ is individually rational, $S=\{i_1,\ldots,i_k\}$ weakly blocks $x$ through the allocation $z$ defined by
\[  z_i \equiv \left \{ \begin{array}{ll}
                w_i & \mbox{if $i\in S$;}\\
                x_i & \mbox{if $i\not \in S$}
                           \end{array}
\right. \]
and for each $l=1,\ldots,k$ (mod $k$), $z_{i_l}=i_{l+1}$ and $(i_l,i_{l+1})\in E$.
This proves the statement for the strong core. The statements for the core and the Wako-core follow similarly.
\end{proof}
\smallskip

An individually rational allocation $x$ is a \emph{maximum size allocation} if for each individually rational allocation $z$, 
$|\{i\in N : x_i\neq i\}| \geq |\{i\in N : z_i\neq i\}|$. Below we provide an example to illustrate the three cores and maximum size allocation.

\begin{example}\label{example:ill} 
{\rm Let $N=\{1,\ldots,6\}$ and let preferences be given by Table~\ref{illpref}, where each agent's own object and all her unacceptable objects are not displayed. For instance, agent 1 is indifferent between objects 2 and 3, and strictly prefers both objects to object 5. 
\begin{figure}[H]
\begin{floatrow}

\capbtabbox{%

\begin{tabular}{c|c|c|c|c|c}
1&2&3&4&5&6\\
\hline
2,3& 1 &2 &3 &2&1\\ 
5&3&4&2&6& \\
\end{tabular}

}{%
  \caption{Preferences}\label{illpref}%
}

\hspace*{-0.5cm}
\ffigbox{%

  \begin{tikzpicture}[->, >=stealth',scale = 0.2, inner sep=2mm]
\tikzstyle{nd}=[circle,draw,fill=black!10,inner sep=0pt,minimum size=6mm]
 \node[nd] (1) at (0,0)  {$1$};
 \node[nd] (2) at (0,-8) {$2$};
 \node[nd] (3) at (-8,0) {$3$};
 \node[nd] (4) at (-8,-8) {$4$};
 \node[nd] (5) at (8,0) {$5$};
 \node[nd] (6) at (8,-8) {$6$};
 \path[every node/.style={font=\sffamily\footnotesize}]
    (1) edge[bend right=15, very thick] (2) 
    (1) edge[very thick] (3)
    (1) edge[darkgray] (5)
    
    (2) edge[bend right=15, very thick] (1) 
    (2) edge[bend right=15, darkgray] (3) 
    
    (3) edge[bend right=15,  very thick] (2) 
    (3) edge[bend right=15,darkgray] (4) 
    
    (4) edge[bend right=15, very thick] (3) 
    (4) edge[darkgray] (2) 
    
    (5) edge[very thick] (2) 
    (5) edge[darkgray] (6) 
    
    (6) edge[very thick] (1) 
    ;
 \end{tikzpicture}

}{%
  \caption{Acceptability graph}\label{illgraph}%
}
\capbtabbox{%
\begin{tabular}{l}
     $x^a=\{(1,3,2)\}$ \\
     $x^b=\{(1,2),(3,4)\}$ \\
     $x^c=\{(1,5,2), (3,4)\}$ \\
     $x^d=\{(1,3,4,2)\}$\\
     $x^e=\{(1,5,6), (2,3,4)\}$ 
\end{tabular}
}{%
  \caption{Allocations}\label{illall}%
}

\end{floatrow}
\end{figure}
\noindent Figure~\ref{illgraph} displays the induced acceptability graph.\footnote{Throughout the paper, self-loops are omitted from the acceptability graphs in the examples.} Here, a thick edge denotes the most preferred object(s) and a thin edge denotes the second most preferred object (if any).

Consider the allocations defined in Table~\ref{illall}. For instance, $x^d$ (in cycle-notation, but without self-cycles) is the allocation $x^d=(x^d_1,x^d_2,x^d_3,x^d_4,x^d_5,x^d_6)=(3,1,4,2,5,6)$. It can be verified that $x^a$ is the unique strong core allocation, $x^a$ and $x^b$ are the competitive allocations, while $x^a$, $x^b$, $x^c$, and $x^d$ form the core. 
Hence, the strong core is a singleton and a proper subset of the set of competitive allocations, while the latter set is also a proper subset of the core. Finally, $x^e$ is the unique maximum size allocation and does not pertain to the core.
\hfill $\diamond$
}
\end{example}
\smallskip

Shapley and Scarf \cite{SS1974} (see also page~135, Roth and Postlewaite \cite{RP1977}) showed that the set of competitive allocations is non-empty and coincides with the set of allocations that are obtained through David Gale's Top Trading Cycles algorithm, which is discussed in the next subsection.\footnote{If preferences are not strict, then the Top Trading Cycles algorithm is applied to the preference profiles that can be obtained by breaking ties in all possible ways.} Roth and Postlewaite \cite{RP1977} showed that if preferences are strict, then there is a unique strong core allocation which coincides with the unique competitive allocation. In general, when preferences are not strict, the strong core can be empty (see, e.g., Footnote~\ref{wakoexample}) or contain more than one allocation (see, e.g., Example~\ref{example:ill}).
However, Wako \cite{Wako1984} showed that the strong core is always a subset of the set of competitive allocations.\footnote{Wako \cite{Wako1991} showed that the strong core coincides with the set of competitive allocations if and only if any two competitive allocations are welfare-equivalent. Hence, whenever the set of competitive allocations is a singleton it coincides with the strong core.} Furthermore, it is easy to see that the set of competitive allocations is always a subset of the core (Shapley and Scarf \cite{SS1974}). 
%

If preferences are strict, the unique competitive allocation is Pareto-efficient (because it is in the strong core) and Pareto-dominates any other allocation
(Lemma~1, Roth and Postlewaite \cite{RP1977}); in particular, any other core allocation is Pareto-inefficient.
If preferences are not strict, it is possible that each competitive allocation is Pareto-dominated by some allocation that is not competitive.\footnote{Example~1 in Sotomayor \cite{Sotomayor2005}, which is attributed to Jun Wako, is illustrative: $N=\{1,2,3\}$ with $2 P_1 3 P_1 1$, $1 I_2 3 P_2 2$, $2 P_3 1 P_3 3$. The set of competitive allocations consists of $x=\{(1,2),(3)\}$ and $x'=\{(1),(2,3)\}$, which are Pareto-dominated by core allocations $\{(1,2,3)\}$ and $\{(1,3,2)\}$, respectively. Moreover, $x_1 P_1 x_1'$ and $x_3' P_3 x_3$. The strong core is empty.\label{wakoexample}}

Finally, competitive allocations need not be welfare-equivalent: in fact, different agents can strictly prefer distinct competitive allocations (see, e.g., Footnote~\ref{wakoexample}). However, Wako \cite{Wako1991} showed that all strong core allocations are welfare-equivalent. The latter result also immediately follows from Quint and Wako's \cite{QW2004} algorithm, which is discussed in the next subsection.

\subsection*{Definitions for bounded length exchanges}

Motivated by kidney exchange programmes, here we consider housing markets where the length of allowed exchange cycles in allocations is limited. Assuming that blocking coalitions are subject to the same limitation, the core and strong core can be adjusted straightforwardly (see also \cite{BMcD2010}).
In view of Wako's \cite{Wako1999} result we  similarly adjust the set of competitive allocations by using the (equivalent) Wako-core.

For a housing market $M=(N,R)$, let $k$ be an integer that indicates the maximal allowed length of exchange cycles. An allocation is a $k$-allocation if each exchange cycle has length at most $k$. 
Formally, an allocation $x$ is a \emph{$k$-allocation} if there exists a partition of $N=S_1\cup S_2\cup\dots \cup S_q$ such that for each $p\in\{1,\dots,q\}$, $|S_p|\leq k$ and $\{x_i : i\in S_p \} =S_p$.  
The definition of the three cores can be adjusted accordingly as well.
Specifically, the \emph{$k$-core} consists of the $k$-allocations for which there is no blocking coalition of size at most $k$; the strong \emph{$k$-core} consists of the $k$-allocations for which there is no weakly blocking coalition of size at most $k$; the Wako-\emph{$k$-core} consists of the $k$-allocations that are not antisymmetrically weakly dominated through a coalition of size at most $k$.
Due to the ``nestedness'' of the three blocking notions, it follows that the strong $k$-core is a subset of the Wako $k$-core, and that the Wako $k$-core is a subset of the $k$-core.
It is also easy to verify that, similarly to the unbounded case, for strict preferences the strong $k$-core coincides with the Wako-$k$-core.

To keep notation as simple as possible, whenever the context is clear, we will omit ``$k$'' from $k$-allocation, $k$-core, etc. and instead refer to $k$-housing markets to invoke the above restriction on exchange cycles, blocking coalitions, allocations, and cores.

The absence of blocking coalitions is also called \emph{stability} in the literature, that is widely used especially for bounded length exchanges. In the case of \emph{pairwise exchanges} (i.e., for $k=2$), the problem is equivalent to the so-called \emph{stable roommates problem}, and \emph{stable marriage problem} if the graph is bipartite, as introduced in \cite{GaleShapley1962}. For strict preferences the core, Wako-core, and strong core are all equivalent and they correspond to the set of \emph{stable matchings}. For weak preferences the core and Wako-core are the same, and correspond to \emph{weakly stable matchings}, whilst the strong core corresponds to \emph{strongly stable matchings}. See more about these concepts in \cite{Manlove2013book}.

\subsection{Algorithms to find all strong core allocations}

In this section we describe the TTC algorithm and its extension by Quint and Wako \cite{QW2004} for finding strong core allocations. Our concise and standardised descriptions provide an easy summary of the current state of the art. Moreover, the graphs defined in the algorithms are crucial tools in Section~\ref{sec:RI} where we prove that the strong core ``respects improvement.'' We consider housing markets with strict preferences and weak preferences separately. In the first case the strong core is always a singleton (which consists of the unique competitive allocation), while in the second case it can be empty.

\subsubsection*{Strict preferences}

Let $M=(N,R)$ be a housing market with strict preferences.
We will construct a subgraph $G^{CP}$ of $G$ by using the Top Trading Cycles (TTC) algorithm of David Gale \cite{SS1974}. The node set of $G^{CP}$ is $N$ and its directed edges $E^{CP}=E^C \cup E^P$ are partitioned into two sets  $E^C$ and $E^P$, where $E^C$ will denote the edges in the TTC cycles and $E^P$ will denote a particular subset of edges pointing to more preferred objects.
\smallskip

\noindent \textbf{TTC algorithm, construction of $G^{CP}$}\\
Set $E^C\equiv \emptyset$, $E^P\equiv\emptyset$, and $M_1\equiv M$. Let
$G_1=(N_1,E_1)\equiv (N,E)$ denote the acceptability graph of $M_1$.
We iteratively construct ``shrinking'' submarkets $M_t$ ($t=2,3,\ldots$) whose acceptability graph will be denoted by $G_t=(N_t,E_t)$. Set $t\equiv 1$.
\begin{enumerate}
    \item Let $E^T_t$ be the set of most preferred edges in $G_t$.
    \item Let $c_t$ be a (top trading) cycle in $(N_t,E^T_t)$. Let $C_t$ and $E_t$ denote the node set and edge set of $c_t$, respectively.
    \item Add the edges of $c_t$ to $E^C$, i.e., $E^C\equiv E^C \cup E_t$.
    \item Let $E^T_t(\vec{C_t})$ denote the subset of edges of $E^T_t$ pointing to $C_t$ from outside $C_t$. Formally, $E^T_t(\vec{C_t})\equiv \{(i,j)\in E^T_t: i\in N_t\setminus C_t \mbox{ and } j\in C_t\}$. Add $E^T_t(\vec{C_t})$ to $E^P$, i.e., $E^P\equiv E^P \cup E^T_t(\vec{C_t})$.
    \item If $N_t=C_t$, stop. Otherwise, let $N_{t+1}\equiv N_t\setminus C_t$, denote the submarket $M_{N_{t+1}}$ by $M_{t+1}$, and go to step 1.
\end{enumerate}
When the algorithm terminates the set of cycles in $E^C$ is the unique competitive allocation and hence the unique strong core allocation.

We classify the relation between any two agents through graph $G^{CP}$ as follows. Let $i,j\in N$ with $i\neq j$. Then, exactly one of the following situations holds:
\begin{itemize}
    \item $i$ and $j$ are \emph{independent}: there is no directed path from $i$ to $j$ or from $j$ to $i$;
    \item $i$ and $j$ are \emph{cycle-members}: there is a path from $i$ to $j$ that entirely consists of edges in $E^C$, i.e., $i$ and $j$ are in the same top trading cycle;\footnote{Obviously, in this case there is also a path from $j$ to $i$ that consists of edges in $E^C$.}
    \item $i$ is a \emph{predecessor} of $j$ (and $j$ is a \emph{successor} of $i$): there is a path from $i$ to $j$ in $G^{CP}$ using at least one edge from $E^P$;\footnote{Note that in this case $j$ was removed from the market before $i$. Hence, there is no path from $i$ to $j$ using only edges in $E^C$ and $j$ is not a predecessor of $i$.} or
    \item $j$ is a \emph{predecessor} of $i$ (and $i$ is a \emph{successor} of $j$): there is a path from $j$ to $i$ in $G^{CP}$ using at least one edge from $E^P$.
\end{itemize}

In case $i$ is a predecessor of $j$, we define \emph{the best path} from $i$ to $j$ to be the path from $i$ to $j$ in $G^{CP}$ where at each node $k\neq j$ on the path, the path follows agent $k$'s (unique) most preferred edge in 
$$\{(k,l) \in E^{CP} : \mbox{ there is a path from $l$ to $j$ using edges in $E^{CP}$}\}.$$
Let $p^b(i,j)$ denote the unique best path from $i$ to $j$ in $G^{CP}$.
For each node $k\neq j$ on $p^b(i,j)$, if there are multiple paths from $k$ to $j$, then $p^b(i,j)$ follows the edge that points to the object that is part of the earliest top trading cycle.

%

\subsubsection*{Weak preferences}

Let $M=(N,R)$ be a housing with weak preferences. 
We will now describe the efficient algorithm of Quint and Wako \cite{QW2004} for finding a strong core allocation  whenever there exists one. We use the simplified interpretation of Cechl\'arov\'a and Fleiner \cite{Cechlarova2010} and construct a subgraph $G^{SP}$ of $G$ with node set $N$ and edge set $E^{SP}\equiv E^S\cup E^P$, which will be useful for our later analysis. 

We first recall two definitions. A \emph{strongly connected component} of a directed graph is a subgraph where there is a directed path from each node to every other node. An \emph{absorbing set} is a strongly connected component with no outgoing edge. Note that each directed graph has at least one absorbing set.\smallskip

\noindent \textbf{Quint-Wako algorithm, construction of $G^{SP}$}\\
Set $E^S\equiv \emptyset$, $E^P\equiv\emptyset$, and $M_1=M$. Let
$G_1=(N_1,E_1)\equiv (N,E)$ denote the acceptability graph of $M_1$. We iteratively construct ``shrinking'' submarkets $M_t$ ($t=2,3,\ldots$) whose acceptability graph will be denoted by $G_t=(N_t,E_t)$. Set $t\equiv 1$.
\begin{enumerate}
    \item Let $E^T_t$ be the set of most preferred edges in $G_t$.
    \item Let $S_t$ be an absorbing set in $(N_t,E^T_t)$.
    Let $N_t(S_t)$ and $E^T_t(S_t)$ denote the node set and edge set of $S_t$.
    \item Add the edges of $S_t$ to $E^S$, i.e., $E^S\equiv E^S \cup E^T_t(S_t)$.
    \item Let $E^T_t(\vec{S_t})$ denote the subset of edges of $E^T_t$ pointing to $N_t(S_t)$ from outside $N_t(S_t)$. Formally, $E^T_t(\vec{S_t})\equiv \{(i,j)\in E^T_t: i\in N_t\setminus N_t(S_t) \mbox{ and } j\in N_t(S_t)\}$. Add $E^T_t(\vec{S_t})$ to $E^P$, i.e., $E^P\equiv E^P \cup E^T_t(\vec{S_t})$.
    \item If $N_t=N_t(S_t)$, stop. Otherwise, let $N_{t+1}\equiv N_t\setminus N_t(S_t)$, denote the submarket $M_{N_{t+1}}$ by $M_{t+1}$, and go to step 1.
\end{enumerate}

Quint and Wako \cite{QW2004} proved that there is a strong core allocation for $M$ if and only if for each absorbing set $S_t$ defined in the above algorithm there exists a cycle cover, i.e., a set of cycles covering all the nodes of $S_t$. Finding a cycle cover, if one exists, can be done with the classical Hungarian method \cite{kuhn1955hungarian} for finding a perfect matching for the corresponding bipartite graph where the objects are on one side, the agents are on the other side, and there is an undirected arc between an object-agent pair if the object is among the agent's most preferred objects (which might include her own object). We refer to \cite{Abraham2007}, \cite{QW2004}, and \cite{Cechlarova2010} for further details on this reduction.

\begin{remark}\label{remark:strongcore}{\rm
If for each absorbing set $S_t$ defined in the above algorithm there exists a cycle cover, then the set of cycle covers (one cycle cover for each absorbing set) constitutes a strong core allocation. Conversely, as shown in the proof of Theorem~5.5 in Quint and Wako \cite{QW2004}, each strong core allocation can be written as a set of cycle covers (one for each absorbing set $S_t$).
Hence, if the strong core is non-empty, all its allocations can be obtained by selecting all possible cycle covers in the algorithm.}
\hfill $\diamond$
\end{remark}

\begin{remark}\label{remark:welfare}{\rm
In the Quint-Wako algorithm, each agent obtains the same welfare at any two cycle covers in which she is involved (because the agent is indifferent between any two of her outgoing edges in an absorbing set). Together with Remark~\ref{remark:strongcore}, this immediately proves Theorem~2(2) in Wako \cite{Wako1991}, which states that all strong core allocations are welfare-equivalent.}
\hfill $\diamond$
\end{remark}

We classify the relation between any two agents through graph $G^{SP}$ as follows. Let $i,j\in N$ with $i\neq j$. Then, exactly one of the following situations holds:
\begin{itemize}
    \item $i$ and $j$ are \emph{independent}: there is no directed path from $i$ to $j$ or from $j$ to $i$;
    \item $i$ and $j$ are \emph{absorbing set members}: there is a path from $i$ to $j$ that entirely consists of edges in $E^S$, i.e., $i$ and $j$ are in the same absorbing set;\footnote{Obviously, in this case there is also a path from $j$ to $i$ that consists of edges in $E^S$.}
    \item $i$ is a \emph{predecessor} of $j$ (and $j$ is a \emph{successor} of $i$): there is 
     a path from $i$ to $j$ in $G^{SP}$ using at least one edge from $E^P$;\footnote{Note that in this case $j$ was removed from the market before $i$. Hence, there is no path from $i$ to $j$ using only edges in $E^S$ and $j$ is not a predecessor of $i$.} or
    \item $j$ is a \emph{predecessor} of $i$ (and $i$ is a \emph{successor} of $j$): there is 
    a path from $j$ to $i$ in $G^{SP}$ using at least one edge from $E^P$.
\end{itemize}

In case $i$ is a predecessor of $j$, a path from $i$ to $j$ in $G^{SP}$ is said to be a best path from $i$ to $j$ if at each node $k\neq j$ on the path, the path follows one of agent $k$'s most preferred edges in 
$$\{(k,l) \in E^{SP} : \mbox{ there is a path from $l$ to $j$ using edges in $E^{SP}$}\}.$$
Let $P^b(i,j)$ denote the \emph{set of best paths} from $i$ to $j$ in $G^{SP}$.

\section{Respecting Improvement}\label{sec:RI}

Let $R,\tilde{R}$ be two preference profiles over objects $N$. Let $i\in N$. We say that \emph{$\tilde{R}$ is an improvement for $i$ with respect to $R$} if
\begin{itemize*}
\item[(1).] $\tilde{R}_i\,=\,R_i$;

\item[(2).] for all $j\neq i$ and all $k$ with $k\, R_j\,  j$,\, $i\, I_j\, k \Longrightarrow i\,\tilde{R}_j\, k$ and $i\, P_j\, k \Longrightarrow i\,\tilde{P}_j\, k$; \mbox{ and}

\item[(3).] for all $j\neq i$ and all $k,l\neq i$ with $k,l\, R_j \, j$,\,\, $k\, R_j\, l \Longleftrightarrow k\, \tilde{R}_j\, l$.
\end{itemize*}
In other words, (1) only agents different from $i$ have possibly different preferences at $\tilde{R}$ and ${R}$, (2) for each agent $j\neq i$, object $i$ can become more preferred than some acceptable objects, and (3) for each agent $j\neq i$ and for each pair of acceptable objects different from $i$, preferences remain unchanged.

As a simple example with $N=\{1,2,3,4,5\}$, let $R$ be any preference profile such that $4\, P_5 \, 1 \, I_5 \, 2 \, I_5 \, 3 \, P_5 5$. Let $\tilde{R}$ be the preference profile where agents $1,2,3$, and $4$ have the same preferences as at $R$ and let $\tilde{R}_5$ be defined by $1\, I_5 4\, P_5 \,  2 \, I_5 \, 3 \, P_5 5$. Then, $\tilde{R}$ is an improvement for agent $1$ with respect to $R$.

\subsection{Strict preferences}

For each profile of strict preferences $R$, let $\tau(R)$ denote the unique competitive allocation (or strong core allocation). We show that $\tau$ \emph{respects improvement} on the domain of strict preferences:

\begin{theorem}\label{theorem:strictRI} For each $i\in N$ and each pair of profiles of strict preferences $R,\tilde{R}$ such that $\tilde{R}$ is an improvement for $i$ with respect to $R$, $\tau_i(\tilde{R})\, R_i \,\tau_i(R)$.
\end{theorem}

\begin{proof}
Let $x=\tau(R)$ and $\tilde{x}=\tau(\tilde{R})$. We can assume that 
there is a unique agent $j\neq i$ with $\tilde{R}_j \neq R_j$ and prove that $\tilde{x}_i R_i x_i$. (If there is more than one such agent, we repeatedly apply the one-agent result to obtain the result.) We can also assume that $i \tilde{R}_j x_j$. (Otherwise $x_j \tilde{P}_j i$ and hence, from the TTC algorithm, $\tilde{x}=x$.)

We distinguish among three cases, depending on the relation between agents $i$ and $j$ in the graph $G^{CP}$ for the market $(N,R)$, i.e., the graph that is obtained in the TTC algorithm for $x$.

\tikzstyle{nd}=[circle, draw, fill = black!10, inner sep=0pt,minimum size=4mm]
\tikzstyle{nde}=[circle, draw, fill= black!10, inner sep=0pt,minimum size=2mm]
\tikzstyle{ed}=[->,shorten >=0.5pt,semithick]   
\tikzstyle{edg}=[->,shorten >=0.5pt,line width = 0.5mm, green!50!black]
    
\begin{figure}
    \centering
    \begin{tabular}{ccc}
          
    \begin{tikzpicture}[ >=stealth']
        \draw (0,0) ellipse (1.2 and 0.6);
        \draw (0,2) ellipse (1.2 and 0.6);
        \draw (1.0,3.5) ellipse (1.0 and 0.5);
        \draw (-0.5,4.8) ellipse (1.0 and 0.5);

        \node[nd] (j) at (1.0,0.3) {\footnotesize $j$};
        \node[nd] (i) at (0.5,1.5) {\footnotesize $i$};
        \node[nde] (e1) at (0,0.6) {};
        \node[nde] (e2) at (0,1.4) {};
        
        \node[nde] (e3) at (0.1, 2.6) {};
        \node[nde] (e4) at (0.4, 3.1) {};
        
        \node[nde] (e5) at (0.4, 3.9) {};
        \node[nde] (e6) at (0, 4.4) {};
        
        \node[nde] (e7) at (-0.4, 2.6) {};
        \node[nde] (e8) at (-0.6, 4.3) {};
        \node[red] (f) at (-1.5, 3.0) {$F(i)$};
        \path[every node/.style={font=\sffamily\footnotesize}]
        (e1) edge[ed] (e2)
        (e3) edge[ed] (e4)
        (e5) edge[ed] (e6)
        (e7) edge[ed] (e8);
        \draw[red, semithick] (-2,1.0) -- (-2,5);
        \draw[red, semithick] (-2,1.0) -- (2.3,1.0);
        \draw[red, semithick] (2.3,1.0) -- (2.3,5);
    \end{tikzpicture}
    &~~~
    \begin{tikzpicture}[ >=stealth']
        \draw (0,0) ellipse (2 and 0.6);
        \draw (1,2) ellipse (1.0 and 0.5);
        \draw (-0.8,3.2) ellipse (1.0 and 0.5);
        \node[green!50!black] (p) at (2.4, 0.5) {\small $p(i,j)$};
        
        \node[nd, fill = green!20, draw = green!50!black] (i) at (1.2,0.5) {\footnotesize $i$};
        \node[green!50!black, fill = white] (dd) at (0.1,0.6) {\footnotesize $\dots$};
        \node[nd, fill = green!20, draw = green!50!black] (j) at (-1.2,0.5) {\footnotesize $j$};
        \node[nde, fill = green!20, draw = green!50!black] (e1) at (0.4,0.6) {};
        \node[nde] (e2) at (0.7,1.5) {};
        
        \node[nde] (e3) at (0.1, 2.2) {};
        \node[nde] (e4) at (-0.2, 2.8) {};
        
        \node[nde, fill = green!20, draw = green!50!black] (e5) at (-0.3, 0.6) {};
        \node[nde] (e6) at (-0.7, 2.7) {};
        
        \node[red] (f) at (0.0, 4.0) {$F^*(N(i,j))$};
        \path[every node/.style={font=\sffamily\footnotesize}]
        (e1) edge[ed] (e2)
        (e3) edge[ed] (e4)
        (e5) edge[ed] (e6)
        (i) edge[edg] (e1)
        (e5) edge[edg] (j);
        \draw[red, semithick] (-2,1.0) -- (-2,5);
        \draw[red, semithick] (-2,1.0) -- (2.3,1.0);
        \draw[red, semithick] (2.3,1.0) -- (2.3,5);
    \end{tikzpicture}
    &~~~
    \begin{tikzpicture}[ >=stealth']
        \draw (0,0) ellipse (0.8 and 0.5);
        \draw (1.2,2.5) ellipse (0.8 and 0.5);
        \draw (-0.5,1.5) ellipse (0.6 and 0.3);
        \draw (-1,2.7) ellipse (0.6 and 0.3);
        
        \draw (0,4) ellipse (0.8 and 0.5);
        \draw (2.3,4.3) ellipse (0.8 and 0.5);
        \draw (2.3,5.6) ellipse (0.6 and 0.3);
    
        \node[nd, fill = green!20, draw = green!50!black] (i) at (0,0.5) {\footnotesize$i$};
        \node[nde] (e2) at (-0.4, 1.2) {};
        \node[nde] (e3) at (-0.5, 1.8) {};
        \node[nde] (e4) at (-0.7, 2.4) {};
        \node[nde] (e5) at (-1.4, 2.5) {};
        \node[nde] (e6) at (-0.4, 2.7) {};
        
        \node[nde, fill = green!20, draw = green!50!black] (l1) at (1.0, 2.0) {};   
        \node[nde, fill = green!20, draw = green!50!black] (l3) at (1.6, 2.05) {};
        \node[nd, fill = green!20, draw = green!50!black] (l) at
        (1.2, 3.0) {\footnotesize$l$};
        \node[nde] (l2) at (0.4, 2.5) {};
        \node[green!50!black, fill = white] (ddl) at (2.12, 2.6) {\footnotesize$\dots$};

        \node[nd, fill = green!20, draw = green!50!black] (j) at
        (0.4, 3.6) {\footnotesize$j$};
        \node[nde] (j1) at (0.8, 4.0) {};
        
        \node[nde] (k1) at (1.5, 4.3) {};
        \node[nd] (k) at (2.3, 3.8) {\footnotesize$k$};
        \node[nde] (k2) at (2.3, 4.8) {};
        
        \node[nde] (u1) at (2.3, 5.3) {};
        
        \node[green!50!black] (p) at (1.4, 1.5) {\small $p^b(i,j)$};
        \node[red] (f) at (1.6, 5.0) {\small $F(k)$};

        \path[every node/.style={font=\sffamily\footnotesize}]
        (i) edge[ed] (e2)
        (e3) edge[ed] (e4)
        (i) edge[ed, bend left = 35] (e5) 
        
        (l2) edge[ed] (e6)
        
        (i) edge[edg] (l1)
        (l) edge[edg] (j)
        (l) edge[ed] (k)
        (l1) edge[edg,bend right = 10] (l3)
        (ddl) edge[edg, bend right = 18] (l)
        (j1) edge[ed] (k1)
        (l3) edge[edg, bend right = 8] (ddl)
        
        (k2) edge[ed] (u1)
        ;
        \draw[red, semithick] (1.1,6) -- (1.1,3.4);
        \draw[red, semithick] (1.1,3.4) -- (3.4,3.4);
        \draw[red, semithick] (3.4,3.4) -- (3.4,6);
        
        
    \end{tikzpicture}\\
    {\sc Case I} & {\sc Case II} & {\sc Case III}
    \end{tabular}

    \caption{Graph $G^{CP}$ (simplified) in the proof of Theorem~\ref{theorem:strictRI}. Each ellipse represents a top trading cycle.} 
    \label{fig:my_label}
\end{figure}
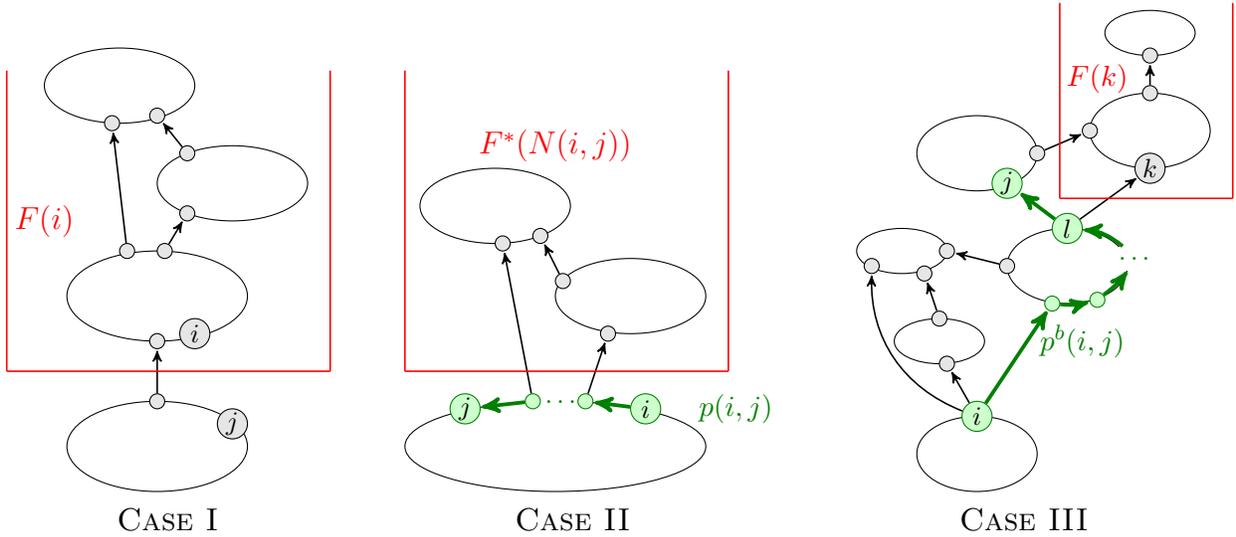

\smallskip

\noindent {\sc Case I:} $i$ and $j$ are independent or $j$ is a predecessor of $i$. Let $F(i)$ be the set of \emph{followers}\footnote{We avoid the use of the usually equivalent nomenclature ``successor'' as the latter term has already a particular (and different) meaning.} of $i$ in the graph $G^{CP}$, where we use the convention $i\in F(i)$. 
Then, $j\not \in F(i)$. In the TTC algorithm for $R$, the agents in $F(i)$ form, among themselves, trading cycles. Since for each agent $k\in F(i)$, $\tilde{R}_k = R_k$, it follows that the trading cycles formed by $F(i)$ in the TTC algorithm for $R$ are also formed in the TTC algorithm for $\tilde{R}$. 
%
%
Hence, $\tilde{x}_i=x_i$.\smallskip

\noindent {\sc Case II:} $i$ and $j$ are cycle-members. Let $C$ be the cycle in the graph $G^{CP}$ that contains $i$ and $j$. Let $p(i,j)$ be the unique path from $i$ to $j$ in the graph $G^{CP}$. Obviously, $p(i,j)$ is part of $C$. Let $N(i,j)$ be the nodes on $p(i,j)$. (So, $i,j\in N(i,j)$.) Let $F^*(N(i,j))$ be the followers outside of $N(i,j)$ that can be reached by some path in $G^{CP}$ that (1) starts from a node in $N(i,j)$ and (2) does \emph{not} contain edges in $C$. Then, the nodes in $F^*(N(i,j))$ constitute trading cycles at $x$. Moreover, the nodes in $F^*(N(i,j))$ are neither predecessors of $j$ nor cycle-members with $j$. Hence, by the same arguments as in Case I, at $\tilde{x}$ the nodes in $F^*(N(i,j))$ constitute the same trading cycles as at $x$. Since $i \tilde{R}_j x_j$, the trading cycle of agent $i$ at $\tilde{x}$ is the cycle that consists of the path $p(i,j)$ and the edge $(j,i)$. Since $p(i,j)$ is part of $C$, it follows that $\tilde{x}_i=x_i$.\smallskip

\noindent {\sc Case III:} $i$ is a predecessor of $j$.
Since for each $k\neq j$, $\tilde{R}_k=R_k$, $i\tilde{R}_j x_j$, and $\tilde{R}_j$ is obtained from $R_j$ by shifting $i$ up, it follows that at some step in the TTC algorithm for $\tilde{R}$, agent $j$ will start pointing to agent $i$ and will keep doing so if and as long as agent $i$ is present. Next, consider the predecessor of $j$ on the path $p^b(i,j)$ (in $G^{CP}$ for the market $(N,R)$), say agent $l$. Let $k\in N$ with $k P_l j$. By definition of $p^b(i,j)$, $k$ and $j$ are not cycle-members nor is $k$ a predecessor of $j$.
From the same arguments as in Case I, the nodes in $F(k)$ (the followers of $k$, where $k\in F(k)$) form among themselves the same trading cycles at $x$ and $\tilde{x}$. Hence, at some step in the TTC algorithm for $\tilde{R}$, agent $l$ will start pointing to agent $j$ and will keep doing so if and as long as agent $j$ is present. We can repeat the same arguments until we conclude that each node in the cycle formed by $p^b(i,j)$ and the edge $(j,i)$ will, at some step, start pointing to its follower and will keep doing so if and as long as the follower is present. Thus, the cycle is a trading cycle at $\tilde{x}$. Let $i'$ be the follower of $i$ in this cycle. Note that in the graph $G^{CP}$, $(i,i')\in E^C$ or $(i,i')\in E^P$. If $(i,i')\in E^C$, then $i'=x_i$, in which case $\tilde{x}_i=i'=x_i$. If $(i,i')\in E^P$, then by definition of $E^P$, $\tilde{x}_i = i' P_i x_i$.
\end{proof}

\subsection{Weak preferences}

For each profile of preferences $R$, let $\mathcal{T}(R)$ denote the set of competitive allocations.

\begin{proposition}\label{proposition:compRI}
For each $i\in N$ and each pair of profiles of preferences $R,\tilde{R}$ such that $\tilde{R}$ is an improvement for $i$ with respect to $R$,
\begin{itemize*}
\item there is $\tilde{x}\in \mathcal{T}(\tilde{R})$ such that for each $x\in \mathcal{T}(R)$, $\tilde{x}_i R_i x_i$; and
\item there is $x\in \mathcal{T}(R)$ such that for each $\tilde{x}\in \mathcal{T}(\tilde{R})$, $\tilde{x}_i R_i x_i$.
\end{itemize*}
\end{proposition}
\noindent In other words, if agent $i$ compares her best allotment at the allocations in $\mathcal{T}(R)$ with her best allotment at the allocations in $\mathcal{T}(\tilde{R})$, then she prefers the latter. Similarly, if agent $i$ compares her worst allotment at the allocations in $\mathcal{T}(R)$ with her worst allotment at the allocations in $\mathcal{T}(\tilde{R})$, then she prefers again the latter. Note that in general there is no competitive allocation where each agent receives her most preferred allotment (among those that are obtained at competitive allocations), i.e., agents do not unanimously agree on the ``best'' competitive allocation (see, e.g., agents 3 and 4 and competitive allocations $x^a$ and $x^b$ in Example~\ref{example:ill}). Nonetheless, Proposition~\ref{proposition:compRI} shows that any individual agent that is systematically optimistic or pessimistic about the specific competitive allocation that is chosen subscribes to the thesis that ``the competitive mechanism'' would respect any of her potential improvements.

\begin{proof}
Let $R^1,R^2,\ldots,R^m$ be the profiles of strict preferences that are obtained from $R$ by breaking ties between acceptable objects in each possible way. 
Similarly, let $\tilde{R}^1,\tilde{R}^2,\ldots,\tilde{R}^p$ be the profiles of strict preferences that are obtained from $\tilde{R}$ by breaking ties between acceptable objects in each possible way, where possibly $p\neq m$.
Then, from Shapley and Scarf \cite{SS1974} (see also page 306 in Wako \cite{Wako1991}), $\mathcal{T}(R)=\{\tau(R^1),\tau(R^2),\ldots,\tau(R^m)\}$ and $\mathcal{T}(\tilde{R})=\{\tau(\tilde{R}^1),\tau(\tilde{R}^2),\ldots,\tau(\tilde{R}^p)\}$. 
It is not difficult to see that for each $R^k$, there is some $\tilde{R}^l$ such that $\tilde{R}^l$ is an improvement for $i$ with respect to $R^k$. Similarly, for each $\tilde{R}^l$, there is some $R^k$ such that $\tilde{R}^l$ is an improvement for $i$ with respect to $R^k$.

We can assume, without loss of generality, that for each $k=1,\ldots,m$, $\tau_i(R^1) R_i \tau_i(R^k)$.
Let $\tilde{R}^l$ be an improvement for $i$ with respect to $R^1$. Then, from Theorem~\ref{theorem:strictRI}, $\tau_i(\tilde{R}^l) R^1_i \tau_i({R}^1)$.
Since $R^1_i = R_i$, $\tau_i(\tilde{R}^l) R_i \tau_i({R}^1)$. This proves the first statement.

We can also assume, without loss of generality, that for each $l=1,\ldots,p$, $\tau_i(\tilde{R}^l) R_i  \tau_i(\tilde{R}^p)$. Let $R^k$ be such that $\tilde{R}^p$ is an improvement for $i$ with respect to $R^k$. Then, from Theorem~\ref{theorem:strictRI}, $\tau_i(\tilde{R}^p) R^k_i \tau_i({R}^k)$.
Since $R^k_i = R_i$, $\tau_i(\tilde{R}^p) R_i \tau_i({R}^k)$. This proves the second statement.
\end{proof}
\bigskip

The following example illustrates Proposition~\ref{proposition:compRI}.

\begin{example}\label{example:cor1} 
{\rm Let $N=\{1,2,3,4,5\}$ and let preferences $R$ and $\tilde{R}$ be given by Tables~\ref{illcor1} and~\ref{illcor1impr}, where only acceptable objects are displayed. Note that $\tilde{R}$ is an improvement for agent $3$ with respect to $R$.
\begin{figure}[H]
\begin{floatrow}
\capbtabbox{%
\begin{tabular}{c|c|c|c|c}
1&2&3&4&5\\
\hline
4& 1 & 4   & 1 & 2\\ 
2& 3,5 & 2 & 4 & 5\\
1& 2 & 3 & & \\
\end{tabular}

\;
}{%
  \caption{$R$}\label{illcor1}%
}
\quad\quad
%
\capbtabbox{%
\begin{tabular}{c|c|c|c|c}
1&2&3&4&5\\
\hline
4& \bf{3} & 4   & 1,\bf{3}& 2\\  
2& 1 & 2 & 4 & 5\\
1& 5 & 3 &  & \\
& 2 &  &  & \\
\end{tabular}
}{%
  \caption{$\tilde{R}$}\label{illcor1impr}%
}
\quad\quad\quad\quad
\capbtabbox{%
\begin{tabular}{l}
     $x^a=\{(1,4),(2,5)\}$ \\
     $x^b=\{(1,4),(2,3)\}$ \\
     $x^c=\{(1,2),(3,4)\}$ 
\end{tabular}
}{%
  \caption{Competitive allocations}\label{illcorall}%
}
\end{floatrow}
\end{figure}
\noindent By applying the TTC algorithm to the strict preferences obtained by breaking all ties in all possible ways, we compute the competitive allocations (Table \ref{illcorall}). In the case of $R$, the two competitive allocations are $x^a$ and $x^b$, and in the case of $\tilde{R}$, the two competitive allocations are $x^b$ and $x^c$.
Hence, both of agent 3's best and worst competitive allotment strictly improve, and at both $R$ and $\tR$ her best allotment is different from her worst allotment.
\hfill $\diamond$
}
\end{example}

In Example~\ref{example:cor1}, for the ``improving agent'' (agent 3), each competitive allotment in the new market is weakly preferred to each competitive allotment in the initial market. However, it is easy to construct housing markets without this feature. 

\smallskip
Next, we turn to the strong core, which in the case of weak preferences is a (possibly proper) subset of the set of competitive allocations. 
Formally, for a profile of preferences $R$, let $\mathcal{SC}(R)$ denote the (possibly empty) strong core of $R$. Since strong core allocations are welfare-equivalent (Remark~\ref{remark:welfare}), we can show that the correspondence $\mathcal{SC}$ \emph{conditionally respects improvement}:

\begin{theorem}\label{theorem:weakRI}
Let $i\in N$. Let $R,\tilde{R}$ be a pair of profiles of preferences such that $\tilde{R}$ is an improvement for $i$ with respect to $R$. If $\mathcal{SC}(R),\mathcal{SC}(\tilde{R})\neq \emptyset$, then
for each $\tilde{x}\in \mathcal{SC}(\tilde{R})$ and for each $x\in \mathcal{SC}(R)$, $\tilde{x}_i R_i x_i$.
\end{theorem}

%

\begin{proof}
Let $x\in\setSC(R)$. From Remark~\ref{remark:welfare} it follows that it is sufficient to show that there exists $\tilde{x}\in\setSC(\tR)$ with $\tilde{x}_i R_i x_i$.
We can assume that there is a unique agent $j\neq i$ with $\tilde{R}_j \neq R_j$. (If there is more than one such agent, we repeatedly apply the one-agent result to obtain the result.) We can also assume that $i \tilde{R}_j x_j$. (Otherwise $x_j \tilde{P}_j i$ and hence, from the Quint-Wako algorithm, $x\in\setSC(R)=\setSC(\tR)$.) 

We distinguish among three cases, depending on the relation between agents $i$ and $j$ in the graph $G^{SP}$ for the market $(N,R)$, i.e., the graph that is generated in the Quint-Wako algorithm to obtain $x$.\smallskip

\noindent {\sc Case I:} $i$ and $j$ are independent or $j$ is a predecessor of $i$. Let $F(i)$ be the followers of $i$ in the graph $G^{SP}$, where we use the convention $i\in F(i)$. Then, $j\not \in F(i)$. In the Quint-Wako algorithm for $R$, the nodes in $F(i)$ are exactly the nodes of a collection of absorbing sets. Since for each agent $k\in F(i)$, $\tilde{R}_k = R_k$, it follows that in the Quint-Wako algorithm for $\tR$ 
the nodes in $F(i)$ are exactly the nodes of the same collection of absorbing sets.
Since $\setSC(\tR)\neq \emptyset$, it follows from Remark~\ref{remark:strongcore} that there exists $\tilde{x}\in \setSC(\tR)$ such that for each agent $k\in F(i)$, $\tilde{x}_k =x_k$. In particular, $\tilde{x}_i=x_i$.\smallskip

\noindent {\sc Case II:} $i$ and $j$ are absorbing set members. Let $S_t$ be the absorbing set that contains $i$ and $j$ in the Quint-Wako algorithm for $R$. Let $(i,l)$ be an edge in $S_t$, where possibly $l=i$.
Let $F^*(N_t(S_t))$ be the followers outside of $N_t(S_t)$ that can be reached by some path in $G^{SP}$ that starts from a node in $N_t(S_t)$. Then, the nodes in $F^*(N_t(S_t))$
are exactly the nodes of a collection of absorbing sets in the Quint-Wako algorithm for $R$. Moreover, the nodes in $F^*(N_t(S_t))$ are neither predecessors of $j$ nor absorbing set members with $j$. Hence, by the same arguments as in Case I, in the Quint-Wako algorithm for $\tR$ the nodes in $F^*(N_t(S_t))$ are again the nodes of the same collection of absorbing sets. Therefore, when the Quint-Wako algorithm is applied to $\tR$, the absorbing set that contains $i$ will again contain $l$. Since $\setSC(\tR)\neq \emptyset$, at each $\tilde{x}\in \setSC(\tR)$, agent $i$ will receive an object $\tilde{x}_i$ such that $\tilde{x}_i \, I_i l$. Since also ${x}_i \, I_i l$, we obtain $\tilde{x}_i \, I_i {x}_i$.
\smallskip

\noindent {\sc Case III:} $i$ is a predecessor of $j$.
Since for each $k\neq j$, $\tilde{R}_k=R_k$, $i\tilde{R}_j x_j$, and $\tilde{R}_j$ is obtained from $R_j$ by shifting $i$ up, it follows that at some step in the Quint-Wako algorithm for $\tilde{R}$, agent $j$ will start pointing to agent $i$ and will keep doing so if and as long as agent $i$ is present. Next, consider the predecessor of $j$ on a best path $p^b(i,j)\in P^b(i,j)$ (in $G^{SP}$ for the market $(N,R)$), say agent $l$. Let $k\in N$ with $k P_l j$. By definition of $p^b(i,j)$, $k$ and $j$ are not absorbing set members nor is $k$ a predecessor of $j$. From the same arguments as in Case I, the nodes in $F(k)$ (the followers of $k$, where $k\in F(k)$) form among themselves the same absorbing sets in the Quint-Wako algorithm for both $R$ and $\tilde{R}$. Hence, during the Quint-Wako algorithm for $\tilde{R}$, at some step agent $l$ will start pointing to agent $j$ and will keep doing so if and as long as agent $j$ is present. We can repeat the same arguments until we conclude that each node in the cycle formed by $p^b(i,j)$ and the edge $(j,i)$ will, at some step, start pointing to its follower and will keep doing so if and as long as the follower is present.
Hence, at some step of the algorithm the cycle formed by $p^b(i,j)$ and the edge $(j,i)$ is part of an absorbing set.
Let $i^b$ denote the follower of agent $i$ in path $p^b(i,j)$.
Since $\setSC(\tR)\neq \emptyset$, at each $\tilde{x}\in \setSC(\tR)$, agent $i$ will receive an object $\tilde{x}_i$ such that $\tilde{x}_i \, I_i i^b$. 
Note that in the graph $G^{SP}$, $(i,i^b) \in E^S$ or $(i,i^b) \in E^P$. If $(i,i^b) \in E^S$, then $i^b I_i x_i$, in which case $\tilde{x}_i I_i x_i$. If $(i,i^b) \in E^P$, then by definition of $E^P$, $i^b R_i x_i$, in which case  $\tilde{x}_i R_i x_i$.
\end{proof}
\smallskip

\begin{corollary}\label{corollary:SC}
For each $i\in N$ and each pair of profiles of strict preferences $R,\tilde{R}$ such that $\mathcal{SC}(R),\mathcal{SC}(\tilde{R})\neq \emptyset$ and $\tilde{R}$ is an improvement for $i$ with respect to $R$,
\begin{itemize*}
\item there is $\tilde{x}\in \mathcal{SC}(\tilde{R})$ such that for each $x\in \mathcal{SC}(R)$, $\tilde{x}_i R_i x_i$; and
\item there is $x\in \mathcal{SC}(R)$ such that for each $\tilde{x}\in \mathcal{SC}(\tilde{R})$, $\tilde{x}_i R_i x_i$.
\end{itemize*}
\end{corollary}


\subsection{Bounded length exchanges}

In this subsection we provide several examples to demonstrate the possible violations of the respecting improvement property (or variants/extensions of the property) in the setting of bounded length exchanges.

\subsubsection*{Pairwise exchanges}

As mentioned in Section~\ref{sec:intro}, the maximisation of the number of pairwise exchanges does not respect improvement. Example~\ref{example:pairwise1} below proves this formally. A consequence is that the priority mechanisms studied by Roth et al.\ \cite{RSU2005} need not be  donor-monotonic if agents' preferences can be non-dichotomous.

\begin{example}\label{example:pairwise1} 
{\rm Let $N=\{1,\ldots,4\}$ and let preferences $R$ be given by Table~\ref{pairwise1pref} and the new preferences $\tilde{R}$ by Table~\ref{pairwise1prefnew}, where the only improvement is that agent 1 becomes acceptable for agent 3. 
\begin{figure}[H]
\begin{floatrow}

\capbtabbox{%

\begin{tabular}{c|c|c|c}
1&2&3&4\\
\hline
2& 1 & 3 &2\\ 
3&4&  & 4\\
1&2& & \\
\end{tabular}

}{%
  \caption{$R$}\label{pairwise1pref}%
}

\hspace*{1.5cm}
\capbtabbox{%
\begin{tabular}{c|c|c|c}
1&2&3&4\\
\hline
2& 1 & \bf{1} &2\\ 
3&4& 3 & 4\\
1&2& & \\

\end{tabular}
}{%
  \caption{$\tilde{R}$}\label{pairwise1prefnew}%
}
\hspace*{1.5cm}
\ffigbox{%

  \begin{tikzpicture}[->, >=stealth',scale = 0.2, inner sep=2mm]
\tikzstyle{nd}=[circle,draw,fill=black!10,inner sep=0pt,minimum size=6mm]
 \node[nd] (1) at (0,0)  {$1$};
 \node[nd] (2) at (0,-8) {$2$};
 \node[nd] (3) at (-8,0) {$3$};
 \node[nd] (4) at (-8,-8) {$4$};
 \path[every node/.style={font=\sffamily\footnotesize}]
    (1) edge[bend right=15, very thick] (2) 
    (1) edge[bend right=15, darkgray] (3)

    (2) edge[bend right=15, very thick] (1) 
    (2) edge[bend right=15, darkgray] (4) 
    
    (3) edge[dashed, bend right=15,  very thick] (1) 

    (4) edge[bend right=15, very thick] (2) 
    
    ;
 \end{tikzpicture}

}{%
  \caption{Acceptability graph}\label{pairwise1graph}%
}
\end{floatrow}
\end{figure}
\noindent Initially, at $R$, there are two ways to maximise the number of pairwise exchanges, namely by picking either of the two-cycles $(1,2)$ and $(2,4)$.
Assume, without loss of generality, that $(1,2)$ is selected. (In case $(2,4)$ is selected, similar arguments can be employed.)
Suppose the discontinuous edge (in~Figure~\ref{pairwise1graph}) is included so that agent 1 “improves” and we obtain $\tilde{R}$.
Then, the unique way to maximise the number of pairwise exchanges is obtained by picking the 2 two-cycles (1,3) and (2,4), which means that agent 1 is strictly worse off than in the initial situation.
%
%
%
\hfill $\diamond$
}
\end{example}

\smallskip

We illustrate that the respecting improvement property can be violated in a weak sense 
, namely the best allotment remains the same, but a worse allotment is created for the improving agent. 

\begin{example}\label{example:pairwise2} 
{\rm Let $N=\{1,\ldots,4\}$ and let preferences $R$ be given by Table~\ref{pairwise2pref} and the new preferences $\tilde{R}$ by Table~\ref{pairwise2prefnew}, where the only improvement is that agent 1 becomes acceptable for agent 3. 
\begin{figure}[H]
\begin{floatrow}

\capbtabbox{%

\begin{tabular}{c|c|c|c}
1&2&3&4\\
\hline
2& 4 & 4 & 3\\ 
3& 1 & 3 & 2\\
1& 2 &   & 4 \\
\end{tabular}

}{%
  \caption{$R$}\label{pairwise2pref}%
}

\hspace*{1.5cm}
\capbtabbox{%

\begin{tabular}{c|c|c|c}
1&2&3&4\\
\hline
2& 4 & \bf{1} & 3\\ 
3& 1 & 4 & 2\\
1& 2 & 3 & 4 \\
\end{tabular}

}{%
  \caption{$\tilde{R}$}\label{pairwise2prefnew}%
}

\hspace*{1.5cm}
\ffigbox{%

  \begin{tikzpicture}[->, >=stealth',scale = 0.2, inner sep=2mm]
\tikzstyle{nd}=[circle,draw,fill=black!10,inner sep=0pt,minimum size=6mm]
 \node[nd] (1) at (0,0)  {$1$};
 \node[nd] (2) at (0,-8) {$2$};
 \node[nd] (3) at (-8,0) {$3$};
 \node[nd] (4) at (-8,-8) {$4$};
 \path[every node/.style={font=\sffamily\footnotesize}]
    (1) edge[bend right=15, very thick] (2) 
    (1) edge[bend right=15, darkgray] (3)

    (2) edge[bend right=15, very thick] (4) 
    (2) edge[bend right=15, darkgray] (1) 
    
    (3) edge[dashed, bend right=15,  very thick] (1) 
    (3) edge[bend right=15, darkgray] (4) 

    (4) edge[bend right=15, darkgray] (2) 
    (4) edge[bend right=15, very thick] (3) 
    
    ;
 \end{tikzpicture}

}{%
  \caption{Acceptability graph}\label{pairwise2graph}%
}
\end{floatrow}
\end{figure}
\noindent Initially, at $R$, the unique (strong) core allocation is $x^a=\{(1,2),(3,4)\}$. Suppose the discontinuous edge (in~Figure~\ref{pairwise2graph}) is included so that agent 1 “improves” and we obtain $\tilde{R}$. Then, another (strong) core solution is created, $x^b=\{(1,3),(2,4)\}$, which is strictly worse for 1.
\hfill $\diamond$
}
\end{example}
\smallskip

The following example illustrates violation of best-Ri property for pairwise exchanges with ties for core and Wako-core.
\begin{example}\label{example:pairwise1} 
{\rm Let $N=\{1,\ldots,4\}$ and let preferences $R$ be given by Table~\ref{pairwise_ties_exampleR}. In the new preferences $\tilde{R}$,  Table~\ref{pairwise_ties_examplenewR}, agent 4 makes an improvement and becomes acceptable for agent 1.

\begin{figure}[H]
\begin{floatrow}

\capbtabbox{%

\begin{tabular}{c|c|c|c}
1&2&3&4\\
\hline
3& 4 & 1,4 & 1\\ 
&  &  & 3\\
&  &   & 4 \\
\end{tabular}

}{%
  \caption{$R$}\label{pairwise_ties_exampleR}%
}

\hspace*{1.5cm}
\capbtabbox{%

\begin{tabular}{c|c|c|c}
1&2&3&4\\
\hline
3& 4 & 1,4 & 1\\ 
\bf{4}&  &  & 3\\
&  &   & 4 \\
\end{tabular}

}{%
  \caption{$\tilde{R}$}\label{pairwise_ties_examplenewR}%
}

\hspace*{1.5cm}
\ffigbox{%

  \begin{tikzpicture}[->, >=stealth',scale = 0.2, inner sep=2mm]
\tikzstyle{nd}=[circle,draw,fill=black!10,inner sep=0pt,minimum size=6mm]
 \node[nd] (1) at (0,0)  {$1$};
 \node[nd] (2) at (8,0) {$2$};
 \node[nd] (3) at (0,-8) {$3$};
 \node[nd] (4) at (8,-8) {$4$};
 \path[every node/.style={font=\sffamily\footnotesize}]
    (1) edge[bend right=15, line width = 0.7mm] (3) 
    (1) edge[bend right=15, dashed, line width = 0.4mm] (4)
    
    (2) edge[bend right=15, line width = 0.7mm] (4) 
    
    (3) edge[bend right=15, line width = 0.7mm] (1)
    (3) edge[bend right=15, line width = 0.7mm] (4)
    
    (4) edge[bend right=15, line width = 0.1mm] (2)
    (4) edge[bend right=15, line width = 0.4mm] (3)
    (4) edge[bend right=15, line width = 0.7mm] (1)
    ;
    
\end{tikzpicture}

}{%
  \caption{Acceptability graph}\label{pairwise_ties_graph}%
}
\end{floatrow}
\end{figure}

For original preferences $R$ there exists two (Wako-) core allocations $x^a = \{(3,4)\}$ and $x^b = \{(1,3), (2,4)\}$, i.e.,  $\mathcal{C}(R) = \mathcal{WC}(R) = \{x^a, x^b\}$. The best allotment for agent 4 agent 3, i.e. allocation $x^a$ is to be chosen.

For preferences $\tilde R$,  newly formed cycle $(1,4)$ is blocking for allocation $x^a$, while $(1,3)$ is blocking for allocation $(1,4)$. Then $\mathcal{C}(\tilde R) = \mathcal{WC}(\tilde R) = \{x^b\}$, hence the improvement is not respected.
}
\end{example}

\subsubsection*{Three-way exchanges}

The following example exhibits three housing markets where for each housing market the three cores coincide (and are non-empty).
Subsequently, we will use the example to show that the three cores do not respect improvement when the maximal allowed length of exchange cycles is 3.

For a housing market $(N,R)$, let, with a slight abuse of notation, $\mathcal{SC}(R)$, $\mathcal{WC}(R)$, and $\mathcal{C}(R)$ denote the (possibly empty) strong core, Wako-core, and core of $R$, respectively.

%
%

\begin{example}
\label{example:length3}
{\rm  Throughout the example we focus on the core. However, since all blocking arguments can be replaced by weak blocking arguments, all statements also hold for the strong core, and hence also for the Wako-core. Let $N=\{1,\ldots,10\}$ be the set of agents. We consider three different housing markets that only differ in preferences. First, consider the housing market $(N,R)$, or simply $R$ for short, with the following ``cyclic'' strict preferences (where unacceptable objects are not displayed):
%
\vspace*{-0.1cm}
\begin{center}
preferences $R$\quad\quad 
\begin{tabular}{c|c|c|c|c|c|c|c|c|c}
1&2&3&4&5&6&7&8&9&10\\
\hline
2 & 3 & 4 & 5 & 6 & 7 & 8 & 9 & 10 & 1\\ 
10& 1 & 2 & 3 & 4 & 5 & 6 & 7 & 8  & 9\\
1 & 2 & 3 & 4 & 5 & 6 & 7 & 8 & 9 & 10\\
\end{tabular}
\end{center}
\vspace*{-0.1cm}
Since only directly neighbouring objects (and one's own object) are acceptable, it follows that the only exchange cycles where each agent is assigned an acceptable object are the 10 self-cycles and the 10 two-cycles $(i,i+1)$ (mod 10) where agents $i$ and $i+1$ swap their objects.\footnote{So, the core coincides with the set of stable matchings of the corresponding ``roommate problem'' (Gale and Shapley, 1962).} The core $\mathcal{C}(R)=\{x^a,x^b\}$ consists of the following two allocations:
\begin{eqnarray*}
x^{a} & = & \{(1,2), (3,4), (5,6), (7,8), (9,10)\} \mbox{ and}\\
x^{b} & = &\{(10,1), (2,3), (4,5), (6,7), (8,9)\}.
\end{eqnarray*}

Next, we create an extended housing markets $R^b$ by inserting one three-cycle in $R$. 

Preferences $R^b$ are  provided in Table~\ref{example3waypref_b}, where the changes with respect to $R$ are bold-faced and depicted in Figure~\ref{example3waygraph_b}.

\begin{figure}[H]
\begin{floatrow}

\capbtabbox{%

\begin{tabular}{c|c|c|c|c|c|c|c|c|c}
1&2&3&4&5&6&7&8&9&10\\
\hline
\bf{4} & 3 & 4 & 5 & 6 & 7 & 8 & \bf{1} & 10 & 1\\
2 & 1 & 2 & \bf{8} & 4 & 5 & 6 & 9 & 8  & 9\\
10 & 2 & 3 & 3 & 5 & 6 & 7 & 7 & 9 & 10\\
1 &  &  & 4 &  &  &  & 8 &  & \\
\end{tabular}
}{%
  \caption{Preferences $R^b$}\label{example3waypref_b}%
}

\ffigbox{%

  \definecolor{cadgreen}{rgb}{0.0, 0.42, 0.24}
\begin{tikzpicture}[shorten >=2pt,->, scale=0.65,  >=stealth',transform shape]
    \tikzstyle{vertex}=[circle,fill=black!25,draw, minimum size=17pt,inner sep=0pt]
    
    \foreach \name/\x in {2/36,3/72, 5/144,  6/180, 7/216, 9/288, 10/324}
    {
        \node[vertex] (\name) at (\x:3) {$\name$};
    }
    \foreach \name/\x in {1/0, 4/108,  8/252}
    {
        \node[vertex, fill = green!35] (\name) at (\x:3) {$\name$};
    }
    
    \foreach \i/\j in {2/3,3/4,5/6,6/7,7/8,9/10,10/1} {
             \draw[very thick, bend right=15, line width = 0.8mm] (\i) to (\j);
    }
    \foreach \i/\j in {1/2,2/3,4/5,5/6,6/7,8/9,9/10} {
             \draw[line width = 0.5mm, bend right=15] (\j) to (\i);
    }
    \draw[green!60!black, line width = 0.8mm] (1) to (4);
    \draw[green!60!black, line width=0.5mm,bend right=15] (1) to (2);
    \draw[green!60!black,,bend right=15] (1) to (10);
    \draw[green!60!black, line width=0.8mm,bend right=15] (4) to (5);
    \draw[green!60!black, line width = 0.5mm] (4) to (8);
    \draw[green!60!black, bend right=15] (4) to (3);
    \draw[green!60!black,line width = 0.8mm] (8) to (1);
    \draw[green!60!black, line width = 0.5mm,bend right=15] (8) to (9);
    \draw[green!60!black,bend right=15] (8) to (7);
    
\end{tikzpicture}
}{%
  \caption{Acceptability graph for $R^b$}\label{example3waygraph_b}%
}
\end{floatrow}
\end{figure}

Apart from the earlier mentioned self-cycles and two-cycles, the only additional exchange cycle with only acceptable objects in $R^b$ is $c^{b}=(1,4,8)$.
Allocation $x^{b}$ is in the core of $R^b$ because $c^{b}$ does not block $x^b$: agent $4$ obtains object $8$ in $c^{b}$, which is strictly less preferred to her assigned object at $x^{b}$. In fact, $x^b$ is the unique core allocation of $R^b$. To see this, note first  that $x^a$ is not in the core of $R^b$ as $c^{b}$ blocks it. And second, the only new exchange cycle created in $R^b$, i.e., $c^{b}$, cannot be part of a core allocation, because if it were, then to avoid blocking cycle $(4,5)$, the next two-cycle $(5,6)$ would have to be part of the allocation, in which case $7$ would remain unmatched (i.e., be a self-cycle) and cycle $(6,7)$ would block the allocation. Therefore, $x^b$ is the unique core allocation of $R^b$, i.e., $\mathcal{C}(R^b)=\{x^b\}$.
\hfill $\diamond$
}
\end{example}

Using the above example we can easily prove the following result. For $i\in N$, preferences $R_i$, and a set of allocations $X$, agent $i$'s most preferred allotment in $X$ is her most preferred allotment among those that she receives at the allocations in $X$.

\begin{proposition}\label{proposition:length3}
Suppose the maximum allowed length of exchange cycles is 3. Then, there are 3-housing markets with strict preferences $(N,R)$ and $(N,\tilde{R})$ with
\begin{itemize*}
\item
$X(R)\equiv \mathcal{SC}(R)=\mathcal{WC}(R)=\mathcal{C}(R)\neq \emptyset$ and
\item
$X(\tilde{R})\equiv \mathcal{SC}(\tilde{R})=\mathcal{WC}(\tilde{R})=\mathcal{C}(\tilde{R})\neq \emptyset$
\end{itemize*}
such that for some $i\in N$, $\tilde{R}$ is an improvement for $i$ with respect to $R$ but
\begin{itemize*}
\item $X(\tilde{R}) \subseteq X(R)$,
\item for the unique $x\in X(R) \backslash X(\tilde{R})$ and for each $\tilde{x} \in X(\tilde{R}) \backslash X(R)$,\,\, $x_i\, P_i \tilde{x}_i$. 
\end{itemize*}
\end{proposition}
\begin{proof}
Let $(N,\tilde{R})$ be the 3-housing market with $N=\{1,\ldots,10\}$ and $\tilde{R}=R^b$ from Example~\ref{example:length3}. 
Let $(N,R)$ be the 3-housing market that is obtained from $(N,\tilde{R})$ by making object 1 unacceptable for agent 8.
Obviously, $\tilde{R}$ is an improvement for agent 1 with respect to $R$.
As shown in Example~\ref{example:length3}, $\mathcal{SC}(\tilde{R})=\mathcal{WC}(\tilde{R})=\mathcal{C}(\tilde{R})=\{x^b\}\neq \emptyset$.
One also easily verifies that $\mathcal{SC}(R)=\mathcal{WC}(R)=\mathcal{C}(R)=\{x^a,x^b\}\neq \emptyset$. Finally, agent 1's most preferred allotment in $\mathcal{SC}(R)=\mathcal{WC}(R)=\mathcal{C}(R)$ is object 2, while agent 1's unique (hence, most preferred) allotment in $\mathcal{SC}(\tilde{R})=\mathcal{WC}(\tilde{R})=\mathcal{C}(\tilde{R})$ is object 10. Since agent 1 strictly prefers object 2 to object 10, the result follows.
\end{proof}

\section{Integer Programming Formulations}\label{sec:IP}

In this section we propose novel integer programming (IP) formulations for the core, the set of competitive allocations (the \emph{Wako-core}), and the strong core. First, we propose  models for the unbounded  case, for the three solution concepts. Second, we propose alternative cycle formulations for the Wako-Quint formulations for core and strong core.
Finally, we propose a new formulation for the competitive allocations. 

\subsection{Novel edge-formulations}\label{sec:novelIP}

Let $(N,R)$ be a housing market and $G\equiv G(N,R)=(N,E)$ its acceptability graph.
Since all three cores only contain individually rational allocations, we can restrict attention to the edges of the acceptability graph.
Specifically, with each edge $(i,j)\in E$ we associate a variable $y_{ij}$ as follows: 
\[  y_{ij} = \left \{ \begin{array}{ll}
                1 & \mbox{if agent $i$ receives object $j$;}\\
                0 & \mbox{otherwise.}
                           \end{array}
\right. \]
%
Then, the base model reads as follows:
%
\begin{align}
& 			     & &  \sum_{j:(i,j)\in E} y_{ij} 	  						 =   1 				& &\forall i\in N \label{eq:EdgeMostOne} \\
&  & &  \sum_{j:(j,i)\in E} y_{ji}  						 = 1	& & \forall i\in N  \label{eq:EdgeInOut} \\
&			     & &  y_{ij} \in \{0,1\}						\label{eq:Edge01}								& & \forall (i,j)\in E
\end{align}
Constraints~\eqref{eq:EdgeMostOne} guarantee that agent $i$ receives exactly one (acceptable) object (possibly her own). Constraints~\eqref{eq:EdgeInOut} guarantee that object $i$ is given to exactly one agent.
Each vector ${(y_{ij})}_{(i,j)\in E}$ that satisfies \eqref{eq:EdgeMostOne}, \eqref{eq:EdgeInOut}, and \eqref{eq:Edge01} yields an allocation $x$ defined by $x_i=j$ if and only if $y_{ij}=1$. Moreover, each allocation can be obtained in this way. So, there is a one-to-one correspondence between allocations and vectors that satisfy \eqref{eq:EdgeMostOne}, \eqref{eq:EdgeInOut}, and \eqref{eq:Edge01}.

We introduce for each $i\in N$ an additional integer variable $p_i$ that represents the price of object $i$.
%
\begin{align}
&			     & &  p_i\in\{1,\dots, n\}  						\label{eq:price}								& & \forall i\in N
\end{align}

In what follows we give our IP formulations for the general case of weak preferences and explain how they can be simplified for strict preferences. We tackle the core, the set of competitive allocations, and the strong core (in this order), by subsequently adding constraints. Given an allocation $x$, we say that $x$ \emph{dominates} an edge $(i,j)$ in the acceptability graph $G$ if agent $i$ weakly prefers her allotment $x_i$ to object $j$, i.e., $x_i\, R_i\, j$.  

\subsubsection*{IP for the core}

It follows from Lemma~\ref{lemma:cycles} that an allocation $x$ is in the core if and only if each cycle in $G$ contains an edge that is dominated by $x$. Or equivalently, there exists no cycle in $G$ that consists of undominated edges. Note that the undominated edges form a cycle-free subgraph if and only if there is a topological order of the objects. The existence of this topological order is equivalent to the existence of prices of the objects such that for each undominated edge $(i,j)$, $p_i<p_j$. Therefore, an allocation $x$ is in the core if and only if there exist prices ${(p_i)}_{i\in N}$ such that 
\begin{equation}\label{domprices}
(i,j)\in E \mbox{ is not dominated by } x \,\, \Longrightarrow \,\, p_i< p_j.\tag{*}
\end{equation}
Thus, core allocations are characterised by constraints \eqref{eq:EdgeMostOne}--\eqref{eq:price} together with \eqref{eq:unbounded_stable} below:

\begin{align}
&			     & &  p_i+1 \leq p_j + n\cdot \sum_{k: kR_ij} y_{ik}  						\label{eq:unbounded_stable}								& & \forall (i,j)\in E
\end{align}

\begin{proposition}\label{proposition:core}
Let $x$ be an allocation. Let $y$ be the corresponding vector that satisfies \eqref{eq:EdgeMostOne}, \eqref{eq:EdgeInOut}, and \eqref{eq:Edge01}. Allocation $x$ is in the core if and only if there are prices ${(p_i)}_{i\in N}$ such that \eqref{eq:price} and \eqref{eq:unbounded_stable} hold.
\end{proposition}
\begin{proof}
First observe that for each $(i,j)\in E$,
\begin{align}
\mbox{ $(i,j)$ is dominated by $x$ } &  \Longleftrightarrow \,\,  x_i\, R_i \, j \nonumber \\ 
& \Longleftrightarrow  \,\, \mbox{there is } k\in N \mbox{ with } k\, R_i \, j \mbox{ and } y_{ik}=1 \nonumber \\
& \Longleftrightarrow  \,\, \sum_{k: kR_ij} y_{ik}=1.\tag{**} \label{keyobs}
\end{align}

Suppose $x$ is in the core. Then, there exist prices ${(p_i)}_{i\in N}$ that satisfy \eqref{eq:price} and \eqref{domprices}. We verify that \eqref{eq:unbounded_stable} holds. Let $(i,j)\in E$. If $(i,j)$ is not dominated by $x$, then \eqref{eq:unbounded_stable} follows immediately from \eqref{domprices}. Suppose $(i,j)$ is dominated by $x$. From \eqref{keyobs}, $\sum_{k: kR_ij} y_{ik}=1$. Hence, 
$$
  p_i+1 \leq n + 1 \leq p_j + n = p_j + n\cdot \sum_{k: kR_ij} y_{ik}.
$$

Suppose that there exist prices ${(p_i)}_{i\in N}$ such that \eqref{eq:price} and \eqref{eq:unbounded_stable} hold. We verify that \eqref{domprices} holds. Let $(i,j)\in E$ and suppose it is not dominated by $x$. From \eqref{keyobs}, $\sum_{k: kR_ij} y_{ik}=0$. Hence, from \eqref{eq:unbounded_stable}, $p_i+1\leq p_j + n\cdot 0$, i.e., $p_i<p_j$.
\end{proof}

\subsubsection*{IP for the set of competitive allocations, i.e., the Wako-core}

The set of competitive allocations is characterised by constraints \eqref{eq:EdgeMostOne}--\eqref{eq:unbounded_stable} together with \eqref{eq:unbounded_ce} below:
\begin{align}
&			     & &  p_i\leq p_j +n \cdot (1- y_{ij})  						\label{eq:unbounded_ce}								& & \forall (i,j)\in E
\end{align}


\begin{proposition}\label{proposition:comp}
Let $x$ be an allocation. Let $y$ be the corresponding vector that satisfies \eqref{eq:EdgeMostOne}, \eqref{eq:EdgeInOut}, and \eqref{eq:Edge01}. Allocation $x$ is competitive if and only if there exist prices ${(p_i)}_{i\in N}$ such that \eqref{eq:price}, \eqref{eq:unbounded_stable}, and \eqref{eq:unbounded_ce} hold. Moreover, if such prices exist, then together with $x$ they constitute a competitive equilibrium.
\end{proposition}
\begin{proof}
    Suppose $x$ is competitive. Let ${(p_i)}_{i\in N}$ be prices such that $(x,p)$ is a competitive equilibrium. Then, \eqref{eq:price} and \eqref{domprices} hold. From the first part of the proof of Proposition~\ref{proposition:core} it follows that \eqref{eq:unbounded_stable} holds. We now prove that \eqref{eq:unbounded_ce} holds as well.
Let $(i,j)\in E$. If $y_{ij}=0$, then immediately $p_i\leq p_j + n = p_j + n \cdot (1-y_{ij})$. If $y_{ij}=1$, then $x_i=j$, and since $(x,p)$ is a competitive equilibrium it follows from Remark~\ref{remark:prices} that $p_i= p_{x_i} =p_j$.

Suppose that there exist prices ${(p_i)}_{i\in N}$ such that \eqref{eq:price}, \eqref{eq:unbounded_stable}, and \eqref{eq:unbounded_ce} hold. We verify that $(x,p)$ is a competitive equilibrium. First, it follows from \eqref{eq:unbounded_ce} that for each $i\in N$, taking $j=x_i$ yields $p_i\leq p_{x_i} + n \cdot (1-1)=p_{x_i}$, i.e., $p_i\leq p_{x_i}$. Hence, from Remark~\ref{remark:prices}, for each $i\in N$,
$p_i= p_{x_i}$.
Second, let $j\in N$ be an object such that $j \, P_i \, x_i$. Then, $(i,j)\in E$ is not dominated by $x$.
From the second part of the proof of Proposition~\ref{proposition:core} it follows that  \eqref{domprices} holds. Hence, we obtain $p_i < p_j$.
\end{proof}

\subsubsection*{IP for the strong core}

The strong core is characterised by constraints \eqref{eq:EdgeMostOne}--\eqref{eq:unbounded_ce} together with \eqref{eq:unbounded_stronglystable} below: 

\begin{align}
&			     & &  p_i\leq p_j + n\cdot \left(\sum_{k: kP_i j} y_{ik}\right)  						\label{eq:unbounded_stronglystable}								& & \forall (i,j)\in E
\end{align}


\begin{proposition}\label{proposition:strongcore}
Let $x$ be an allocation. Let $y$ be the corresponding vector that satisfies \eqref{eq:EdgeMostOne}, \eqref{eq:EdgeInOut}, and \eqref{eq:Edge01}. Allocation $x$ is in the strong core if and only if there exist prices ${(p_i)}_{i\in N}$ such that \eqref{eq:price}, \eqref{eq:unbounded_stable}, \eqref{eq:unbounded_ce}, and \eqref{eq:unbounded_stronglystable} hold. Moreover, if such prices exist, then together with $x$ they constitute a competitive equilibrium.
\end{proposition}
\begin{proof}
Suppose $x$ is in the strong core. By Remark~\ref{remark:strongcore}, $x$ can be obtained in the Quint-Wako algorithm by choosing for each absorbing set in the algorithm a particular cycle cover.
Hence, there exist price ${(p_i)}_{i\in N}$ such that (i) constraints~\eqref{eq:price} are satisfied, (ii) all objects in the same absorbing set have the same price, and (iii) an absorbing set that is processed earlier by the algorithm has a strictly higher associated price (of its objects).
It is easy to verify that $(x,p)$ is a competitive allocation. Hence, from the first part of the proof of Proposition~\ref{proposition:comp} it follows that  \eqref{eq:unbounded_stable} and \eqref{eq:unbounded_ce} hold. Finally, to see that \eqref{eq:unbounded_stronglystable} holds
note that from the definition of the prices it follows that 
(i) if $j R_i x_i$ then $p_i\leq p_j$ and (ii) if $x_i P_i j$ then $p_i \leq n = n (\sum_{k: kP_i j} y_{ik})$.

Suppose that there exist prices ${(p_i)}_{i\in N}$ such that \eqref{eq:price}, \eqref{eq:unbounded_stable}, \eqref{eq:unbounded_ce}, and \eqref{eq:unbounded_stronglystable} hold. It follows from Proposition~\ref{proposition:comp} that $(x,p)$ is a competitive equilibrium. We prove that $x$ is a strong core allocation. Suppose there is a coalition $S$ that weakly blocks $x$ through an allocation $z$. From Lemma~\ref{lemma:cycles} it follows that we can assume, without loss of generality, that $S=\{1,\ldots,r\}$ and that for each $i=1,\ldots,r-1$, $z_i=i+1$, $z_r=1$, and $z_1 P_1 x_1$. Since $x$ is individually rational, $r>1$. Since $(x,p)$ is a competitive equilibrium, $p_1<p_2$. 
Since $3 = z_2 \, R_2 \, x_2$, we have $\sum_{k: kP_2 3} y_{2k}=0$. Hence, from \eqref{eq:unbounded_stronglystable},
$$p_2\leq p_3 + n\cdot \left(\sum_{k: kP_2 3} y_{2k}\right)= p_3.$$
So, $p_2\leq p_3$. By repeatedly applying the same arguments we find $p_2\leq p_3 \leq \cdots \leq p_r\leq p_1$. Since $p_1<p_2$, we obtain a contradiction. Therefore, there is no coalition that weakly blocks $x$. Hence, $x$ is a strong core allocation.
\end{proof}
\smallskip

\begin{remark}
{\rm
We note that in the case of strict preferences, constraints  \eqref{eq:unbounded_stronglystable} are satisfied by any competitive equilibrium $(x,p)$.
To see this note that if $y_{ij}=1$ then \eqref{eq:unbounded_ce} implies \eqref{eq:unbounded_stronglystable}, since $1-y_{ij}=0$, and hence $$p_i\leq p_j + n\cdot (1-y_{ij})=p_j\leq p_j + n\cdot \left(\sum_{k: kP_i j} y_{ik}\right).$$  Otherwise, if $y_{ij}=0$ then \eqref{eq:unbounded_stable} implies \eqref{eq:unbounded_stronglystable}, since for strict preferences $\sum_{k: kP_i j} y_{ik}=\sum_{k: kP_i j} y_{ik}+y_{ij}=\sum_{k: kR_i j} y_{ik}$, and hence $$p_i<p_i+1\leq p_j + n\cdot \left(\sum_{k: kR_i j} y_{ik}\right)= p_j + n\cdot \left(\sum_{k: kP_i j} y_{ik}\right).$$ Therefore, in either case, constraints \eqref{eq:unbounded_stronglystable} are satisfied. This reflects the fact that for strict preferences the strong core is a singleton that consists of the unique competitive allocation.
\hfill $\diamond$
}
\end{remark}

\subsection{Quint and Wako's IP formulations}

To compare our IP formulations with the IP formulations for the core and the strong core given by Quint and Wako \cite{QW2004}, we describe the latter IP formulations using our notation. 

First, for both the core and the strong core, Quint and Wako \cite{QW2004} used the ``basic'' constraints \eqref{eq:EdgeMostOne}, \eqref{eq:EdgeInOut}, and \eqref{eq:Edge01}. We refer to their formulas (9.2), (9.3), (9.4), as well as (8.2), (8.3), (8.4), together with an integrality condition.

Next, to obtain the core Quint and Wako \cite{QW2004} imposed the following additional no-blocking condition (see (9.1) in \cite{QW2004}):
\begin{align}
&			     & &  \sum_{i\in S} \left (\sum_{j: jR_i \pi_i} y_{ij} \right )\geq 1  						\label{eq:QW2004_weak}								& & \forall S\subseteq N, \pi\in \Pi_S
\end{align}

Finally, to obtain the strong core Quint and Wako \cite{QW2004} imposed the following additional no-blocking condition (see (8.1) in \cite{QW2004}):
\begin{align}
&			     & &  \sum_{i\in S} \left (\sum_{j: jP_i \pi_i} y_{ij} + \frac{1}{|S|}\sum_{j: jI_i \pi_i} y_{ij} \right )\geq 1  						\label{eq:QW2004_strong}								& & \forall S\subseteq N, \pi\in \Pi_S,
\end{align}
where $\Pi_S$ is the set of allocations in the submarket $M_S$ (so that $\pi$ is an allocation in $M_S$). 

Constraints (\ref{eq:QW2004_weak}) and (\ref{eq:QW2004_strong}) directly describe that no coalition $S$ can block / weakly block through an allocation $\pi$, respectively. Both sets of constraints are highly exponential (in the number of agents), since they are required not only for all subsets $S$ of $N$, but also for all possible redistributions within each $S$.

%
%

\smallskip

\noindent {\bf{Alternative cycle-formulations}}

\noindent
In view of Lemma~\ref{lemma:cycles}, it is sufficient to impose constraints (\ref{eq:QW2004_weak}) and (\ref{eq:QW2004_strong}) for the cycles of the acceptability graph $G$. Based on this observation and results in \cite{Klimentova2020stable}, we will
describe alternative cycle-formulations for the core and the strong core.
Furthermore, we will provide a new proposition and IP formulation for the Wako-core.





Let $M=(N,R)$ be a housing market. Let $\mathcal{C}$ denote the set of exchange cycles in $G(N,R)$. For a cycle $c\in \mathcal{C}$, let $N(c)$ and $A(c)$ denote the set of nodes and edges in $c$, respectively, and let $|c|$ denote the size/length of $c$. We write $c_i=j$ if agent $i$ receives object $j$ in the exchange cycle $c$, i.e., $(i,j)\in A(c)$.

\begin{proposition}[\cite{Klimentova2020stable}]
An allocation $x$ is in the core if and only if for each cycle $c\in\mathcal{C}$, for some agent $i\in N(c)$, $x_iR_ic_i$.\label{X1}
\end{proposition}

\noindent The corresponding IP constraints, which reduce the constraints \eqref{eq:QW2004_weak} to cycles, are as follows:
\begin{align}
\sum_{(i,j)\in A(c)}\sum_{k:k R_i j} y_{ik} \geq 1 && \forall c\in\mathcal{C}\label{eq:stab_edge}
\end{align}

Next, we describe the alternative cycle-formulation for the strong core. First we focus on the special case of strict preferences.

\begin{proposition}[\cite{Klimentova2020stable}]
Suppose preferences are strict. Then, an allocation $x$ is in the strong core if and only if for each cycle $c\in\mathcal{C}$, $c$ is an exchange cycle in $x$ or for some agent $i\in N(c)$, $x_iP_ic_i$.\label{X2}
\end{proposition}

\noindent Proposition~\ref{X2} leads to the following constraints:
\begin{align}
\sum_{(i,j)\in A(c)} y_{ij} + |c|\cdot\left[ \sum_{(i,j)\in A(c)}\,\, \sum_{k: kP_ij} y_{ik}\right]   \geq |c| && \forall c\in\mathcal{C}
\end{align}

The alternative cycle-formulation for the strong core in the general case (where  preferences can have ties) is as follows.

\begin{proposition}[\cite{Klimentova2020stable}]
An allocation $x$ is in the strong core if and only if for each cycle $c\in\mathcal{C}$,\\ 
(i) $c$ is an exchange cycle in $x$, or\\
(ii) for some agent $i\in N(c)$, $x_iP_ic_i$, or\\
(iii) for each agent $i\in N(c)$, $c_iI_ix_i$.\label{X3} 
\end{proposition}

\noindent The corresponding IP constraints, which reduce the constraints \eqref{eq:QW2004_strong} to cycles, are as follows:
\begin{align}
\sum_{(i,j)\in A(c)} \sum_{k: kI_ij} y_{ik}+|c|\cdot\left[ \sum_{(i,j)\in A(c)}\,\, \sum_{k:kP_i j} y_{ik}\right]   \geq |c| && \forall c\in\mathcal{C}\label{eq:s_stab_edge_ties}
\end{align}

Finally, similarly to the core and strong core, we provide a new alternative characterisation for the Wako-core.

\begin{proposition}\label{prop:Wako-bounded}
An allocation $x$ is in the Wako-core if and only if for each cycle $c\in\mathcal{C}$,\\ 
(i) $c$ is an exchange cycle in $x$, or\\
(ii) for some agent $i\in N(c)$, $x_iP_ic_i$, or\\
(iii) for some agent $i\in N(c)$, $c_iI_ix_i$ and $c_i\neq x_i$. 
\end{proposition}

\noindent The proof of Proposition~\ref{prop:Wako-bounded} is omitted as it is very similar to that of Proposition~\ref{X3} (see \cite{Klimentova2020stable}). Proposition~\ref{prop:Wako-bounded} 
leads to the following constraints, which can be used to find competitive allocations (i.e., allocations in the Wako-core):
\begin{align}
\sum_{(i,j)\in A(c)} y_{ij}+|c|\cdot\left[ \sum_{(i,j)\in A(c)}\,\, \sum_{k: kR_ij, k\neq j} y_{ik}\right]   \geq |c| && \forall c\in\mathcal{C}\label{eq:Wako-core}
\end{align}
To see the correctness of this new formulation, observe that the first term of \eqref{eq:Wako-core} is equal to $|c|$ if condition (i) of Proposition~\ref{prop:Wako-bounded} is satisfied and less than $|c|$ otherwise; and the second term has value at least $|c|$ if condition (ii) or (iii) of Proposition~\ref{prop:Wako-bounded} is satisfied and 0 otherwise. Therefore, constraint \eqref{eq:Wako-core} is satisfied if and only if at least one of the three conditions of Proposition~\ref{prop:Wako-bounded} holds.

\subsection{Bounded length exchanges}

Note that the above cycle-formulations are not very practical due to the exponentially large number of cycles. In fact, this justified the novel IP formulations proposed in Section~\ref{sec:novelIP}. 
However, the cycle-formulations are practical for the case of \emph{bounded} length exchanges.

One easily verifies that Lemma~\ref{lemma:cycles} can be extended to bounded length exchanges in a natural way: the strong core, Wako-core, and core of a $k$-housing market can be defined equivalently by the absence of corresponding blocking \emph{cycles} of size at most $k$. In fact, Klimentova et al.~\cite{Klimentova2020stable} proposed associated IP formulations by adapting constraints \eqref{eq:stab_edge} and \eqref{eq:s_stab_edge_ties} to bounded exchange cycles.
One can similarly adapt constraints \eqref{eq:Wako-core} to obtain an IP formulation for the Wako-core of a $k$-housing market. In our simulations we used the most efficient cycle-edge formulations by Klimentova et al. (see the detailed description in section 3.3 of ~\cite{Klimentova2020stable}).

\section{Computational Experiments}\label{sec:comp}
%

In this section we perform a computational analysis of the models proposed in Section~\ref{sec:IP} and compare them with the models for bounded length exchange cycles, proposed in~\cite{Klimentova2020stable}. Furthermore, we estimate the frequency of violations of the respecting improvement property for all models by computational simulations.
The models are run for both strict and weak preferences and considering two objective functions: maximisation of the size of the exchange (corresponding to the maximisation of the number of transplants in the context of KEPs), denoted by Max$_t$, and maximisation of total weight (where weights can mean the scores given to the corresponding transplants in a KEP, reflecting the qualities of the transplants), denoted by Max$_w$.
For bounded length exchanges, the maximum length considered are $k=2$ and $k=3$.



In Section \ref{sec:results_stable}, we compare the size/weight of the maximum size/weight allocation to the core, competitive and strong core allocations under the same objective. For the unbounded case, we further analyse the 
\emph{price of fairness}: the difference in percentage in the number of transplants of the maximum total weight solution, and the core, competitive and strong core allocations for both objectives, when compared to the maximum size solution.

In Section \ref{sec:results_unstable} we calculate the average number of weakly blocking cycles  for allocations provided by each formulation. By doing so, we give a rough indication of  how far each solution is from the strong core. We complement that analysis with the quantification of the average number of vertices of an instance that may obtain a strictly better allotment in at least one weakly blocking cycle.

Average CPU times required to solve an instance of a given size for each of the formulations in Section \ref{sec:IP} are presented in subsection \ref{sec:time}.



Finally, in subsection~\ref{sec:violRI} we provide results on the frequency of violations of the respecting improvement property for all of our models.

All programs were implemented using Python programming language and tested using Gurobi as optimisation solver~\cite{gurobi}. Tests were executed on a MacMini 8 running macOS version
10.14.3 in a Intel Core i7 CPU with 6 cores at 3.2 GHz with 8GB of RAM.


Test instances were generated with the generator proposed in~\cite{Santos17,fair19} and are available from http: TBA. The number of pairs of an instance ranges from 20 to 150; 50 instances of each size were generated. 
The weights associated to the arcs of the graph were generated randomly within the interval $(0,1)$, and preferences were assigned in accordance with those weights: the higher is the weight of an outgoing arc for a given vertex, the more preferred is the corresponding good. For weak preferences, outgoing arcs with weights within each interval of length $\frac{1}{|V|}$ were considered equally preferable. 







    

\subsection{Impact of stability on the number of transplants}\label{sec:results_stable}

Figure \ref{fig:abs_val_weak} presents average results for the maximum size and maximum weight objectives for weak preferences under different settings: no stability requirements ({Max}), core, competitive and strong core allocations. We refrain ourselves from presenting the results for the case of strict preferences as all curves are similar, except that for strict preferences the competitive and strong core allocations are the same. 

As expected, both the number of transplants and total weight decrease by increasing the number of constraints from {Max} to Core, then to Competitive, and then to Strong Core allocations. The strong core curve is non-monotonic, which is explained by the absence of feasible solutions for several instances. Next to the curve we present the number of instances out of the total 50 where a feasible solution existed. 

\begin{figure}[htbp]

    \centering
    \includegraphics[scale = 0.65]{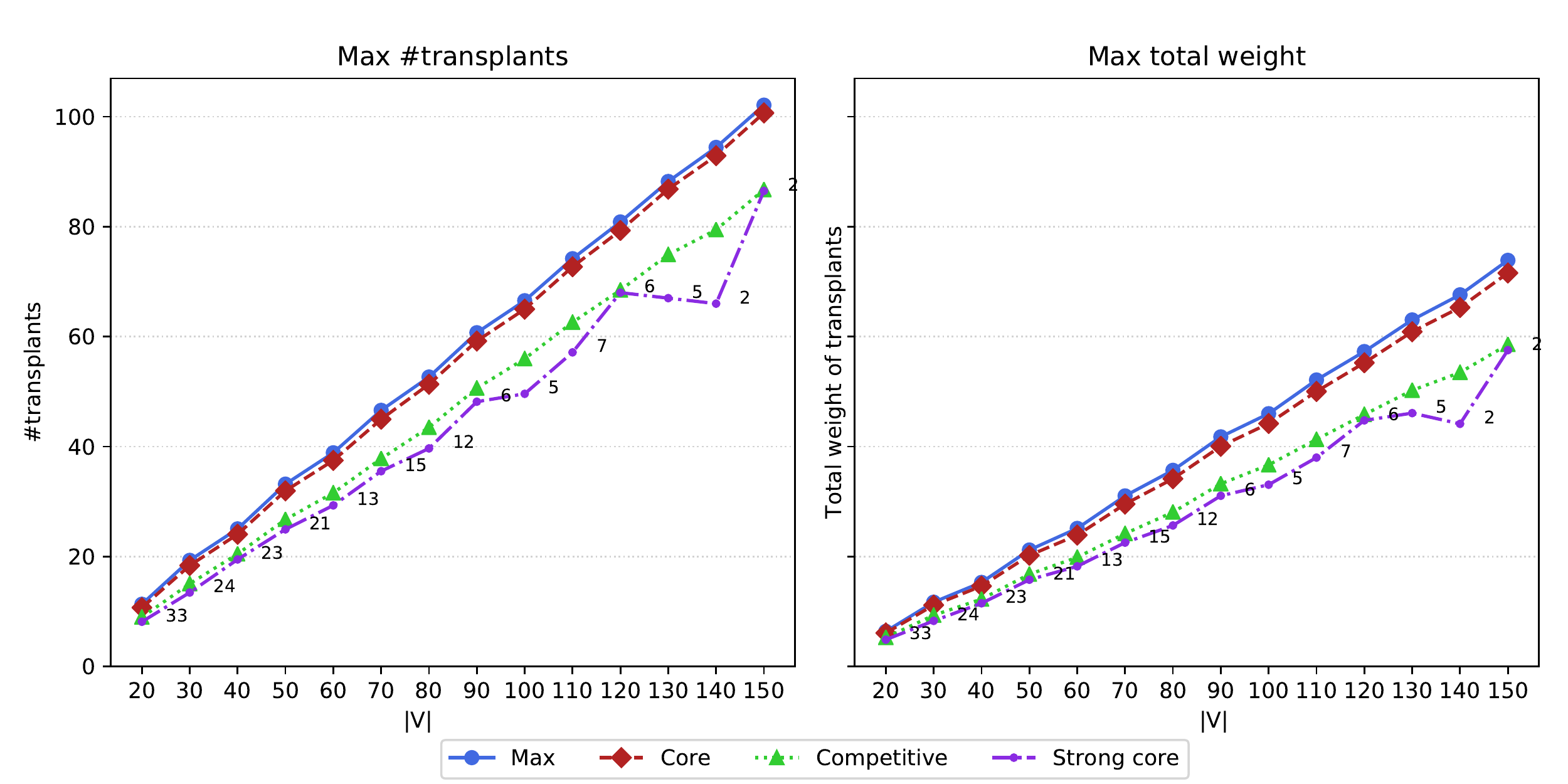}
    \caption{Number of transplants (left) and total weight of transplants (right) for unbounded length and weak preferences. The numbers in the chart reflect the number of instances where a feasible solution existed.}
    \label{fig:abs_val_weak}
\end{figure}

Figure \ref{fig:abs_val_all} makes a similar analysis for the bounded case, when $k = 2$ and $k = 3$. Max$^{k = \infty}$  refers to the unbounded exchange problem, whilst Max$^{k =2}$ and Max$^{k =3}$ correspond to the bounded problem for $k = 2$ and $k = 3$, respectively. The same reasoning is used for the notation associated to the Wako-core \footnote{As mentioned before, for the unbounded case competitive allocations are equivalent to Wako-core, and for bounded case we only have Wako-core.} ({W.-Core}) and strong core ({S.Core)}.     
 For easiness of comparison between the bounded and the unbounded cases, we again plot the two curves from Figure \ref{fig:abs_val_weak}  associated with maximum utility ({Max}) which, in both cases, represent an upper bound for our solutions.
Naturally, the curves associated with  $k = 2$ are dominated by those associated with  $k = 3$.
We can observe that the maximum number of transplants for $k = 3$ and for unbounded $k$ are very similar (see Figure \ref{fig:abs_val_all} (left)).
Notice also that even though some curves overlap and seem identical, there are minor differences among them, except for the case of core and Wako-core allocations for $k=2$, that coincide.
Again, we only present results for weak preferences, as this is the more general case. For strict preferences, for $k = 3$ the curves are similar, 
for $k = 2$, core, competitive and strong core allocations coincide, and the latter two are also the same for unbounded exchanges.


\begin{figure}

    \centering
    \includegraphics[scale = 0.7]{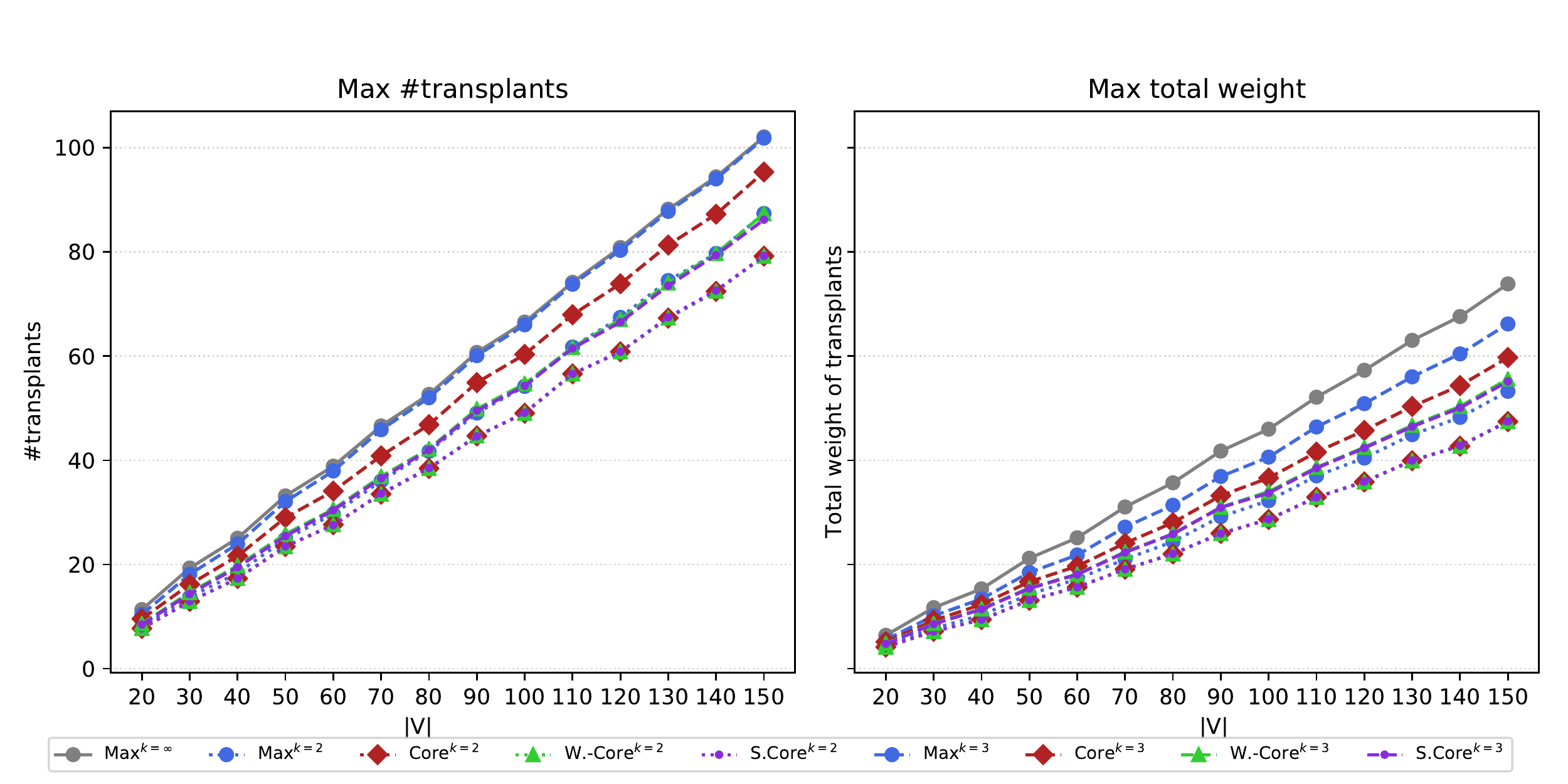}
    \caption{Comparison of the number of transplants (left) and the total weight of transplants (right) for bounded length exchanges ($k=2,3$) and weak preferences. A solid line is used for the unbounded case, dotted lines are used for $k = 2$ and dashed lines for $k = 3$.}
    \label{fig:abs_val_all}
\end{figure}

From a practical point of view an interesting question is to study the impact of stability requirements on the number of transplants achievable. Although KEPs have many other key performance indicators, this is unarguably the most relevant one, as this criterion is optimised as a first objective in all the European KEPs \cite{biroetal2019b}. 
Figure \ref{fig:pf_weak} presents the \emph{price of fairness}, that is difference in percentage in the number of transplants for Max$_w$ allocation, and for Core, Competitive and Strong Core allocations under both objectives, when compared to the maximum number of transplants achievable (Max$_t$).
Subscripts $t$ and $w$ identify the objective functions used for each allocation. 
As shown, 
the price of fairness for competitive and strong core allocations is extremely high, when compared to the core. 
It decreases with problem size for both objective functions and for all allocation models, being slightly higher for the core for the total weight objective (see curve Core$_w$). For the  maximisation of the number of transplants (curve Core$_t$), for instances with more than 50 nodes the reduction is of less than 3\%, decreasing to 1\% for the largest instance. Such result is of major practical relevance as it indicates that with increasing size of the programs one can consider pairs' preferences in the matching with no significant reduction in the number of pairs transplanted.  
\begin{figure}[htbp]
    \centering
    \includegraphics[scale = 0.7]{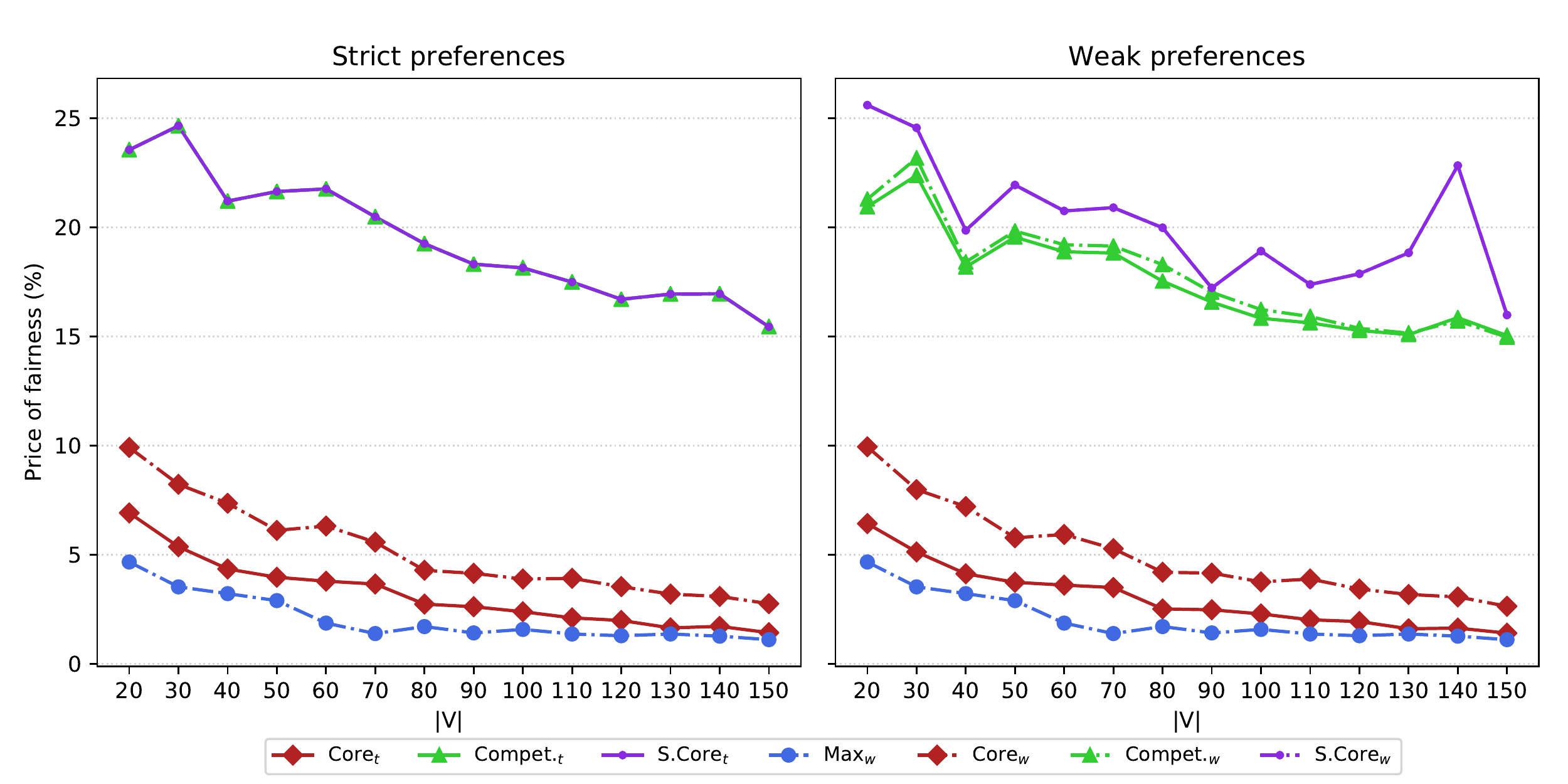}
     \caption{Price of fairness with respect to maximum number of transplants for core, competitive and strong core allocations, with maximum number of transplants and maximum total weight, and solution with maximum total weight; for strict (left) and weak (right) preferences.}
    \label{fig:pf_weak}
\end{figure}

\subsection{Assessing the distance of different solutions from the strong core}\label{sec:results_unstable}

Figure \ref{fig:wb_up2} (left) presents the average number of weakly blocking cycles of size 2 in  Max,  Core,  and Competitive (Wako-core) allocations. We denote the maximum length of the blocking cycles considered by $l$. 
For the bounded case, following the same reasoning as in \cite{Klimentova2020stable}, the figure also reports the minimum average number of weakly blocking cycles for the cases where the strong core does not exist, i.e., for the maximum number of transplants/total weight solution with minimum number of weakly blocking cycles.  Interestingly, when the objective function is the number of transplants, the ``unstability'' of the solutions barely depends on the size of exchanges allowed. The same does not hold for the core, where the number of blocking cycles is considerably smaller for $k = 2$. For this and all the remaining cases, the average number of weakly blocking cycles is very low, in most cases below 1. It is worth to note that the average number of blocking cycles tends to be smaller when the objective is to maximise the total weight (Figure \ref{fig:wb_up2} (right)). A plausible justification for this is that the weights reflect patients preferences and therefore a solution obtained by considering that objective will be closer to a stable solution.
\begin{figure}
    \centering
    \includegraphics[scale = 0.7]{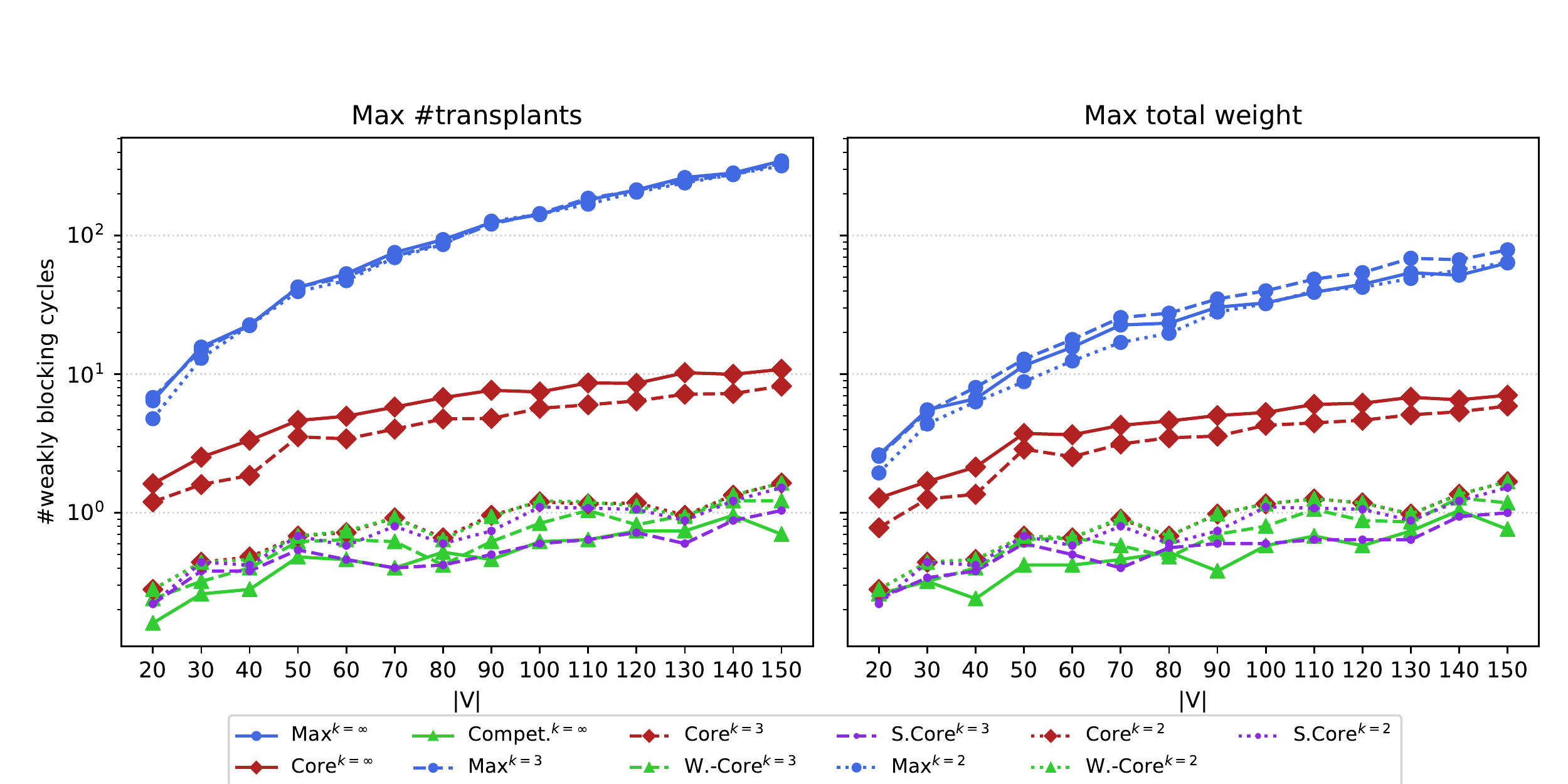}
    \caption{Number of weakly blocking cycles of size ${l=2}$ for solutions with maximum number of transplants (left) and maximum total weight of transplants (right), for unbounded exchanges and exchanges of size up to $k=2$ and $k=3$ for weak preferences.  A solid line is used for the unbounded case, dotted lines are used for $k = 2$ and dashed lines for $k = 3$.}
    \label{fig:wb_up2}
\end{figure}

Figure \ref{fig:wb_up3} presents the same analysis, now considering weakly blocking cycles of size up to 3. Naturally, the solutions for $k = 2$ are excluded from this analysis, as they are fully reflected in Figure \ref{fig:wb_up2}. The conclusions drawn for $l = 2$ remain valid for this case. 
\begin{figure}
\includegraphics[scale = 0.7]{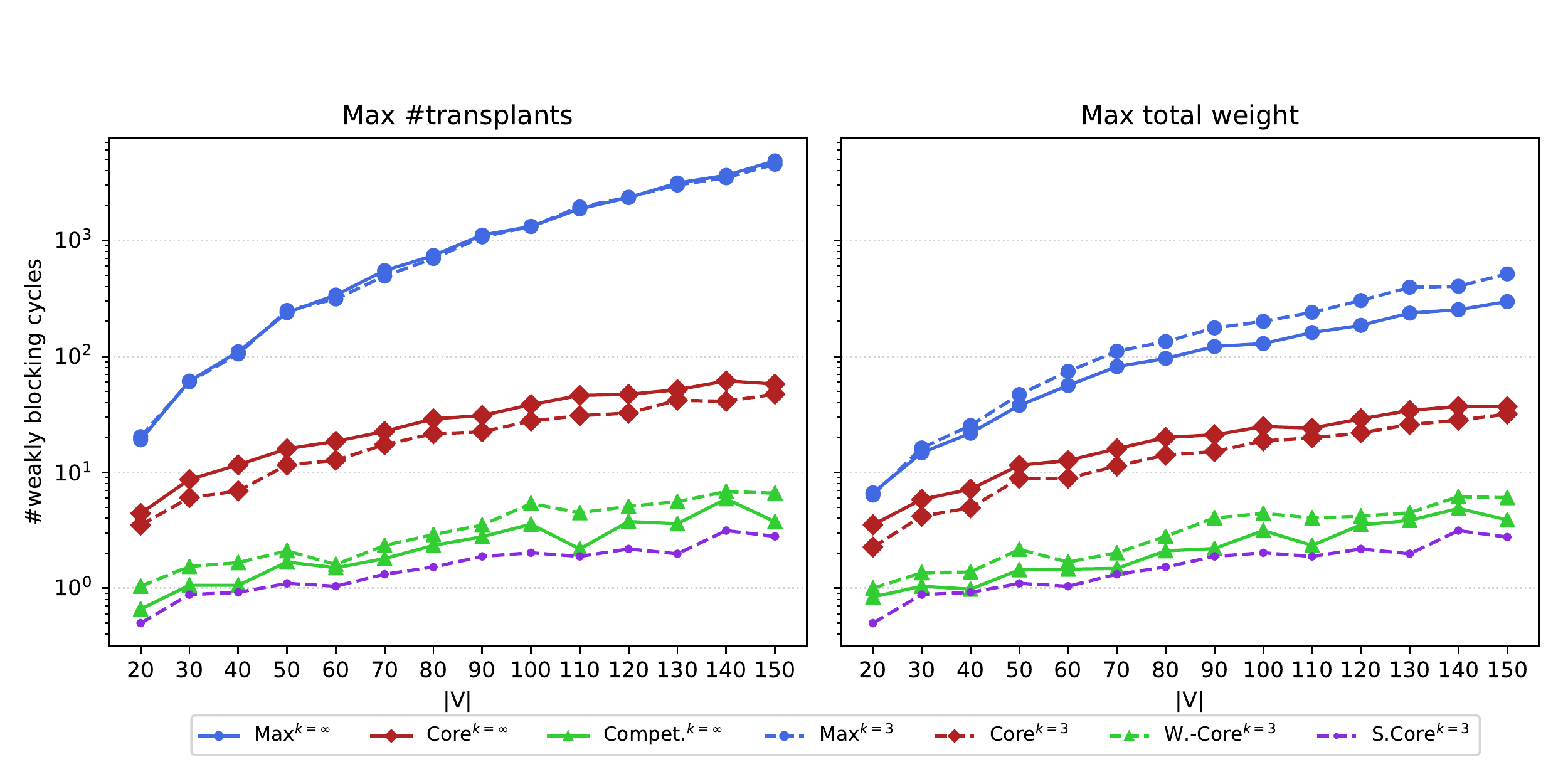}
    \caption{Number of weakly blocking cycles of size ${l=3}$ for solutions with maximum number of transplants (left) and maximum total weight of transplants (right), for unbounded exchanges and exchanges of size up to $k=3$ for weak preferences.}
     \label{fig:wb_up3}
\end{figure}
For the unbounded case, the number of blocking cycles is larger, since one must consider also the cases when $l > 3$.  Figure \ref{fig:wb_up5} provides information on the number of weakly blocking cycles of size up to 4 and up to 5. We do not present results for $l > 5$ as searching for those blocking cycles would exceed our CPU time. 

\begin{figure}
\includegraphics[scale = 0.7]{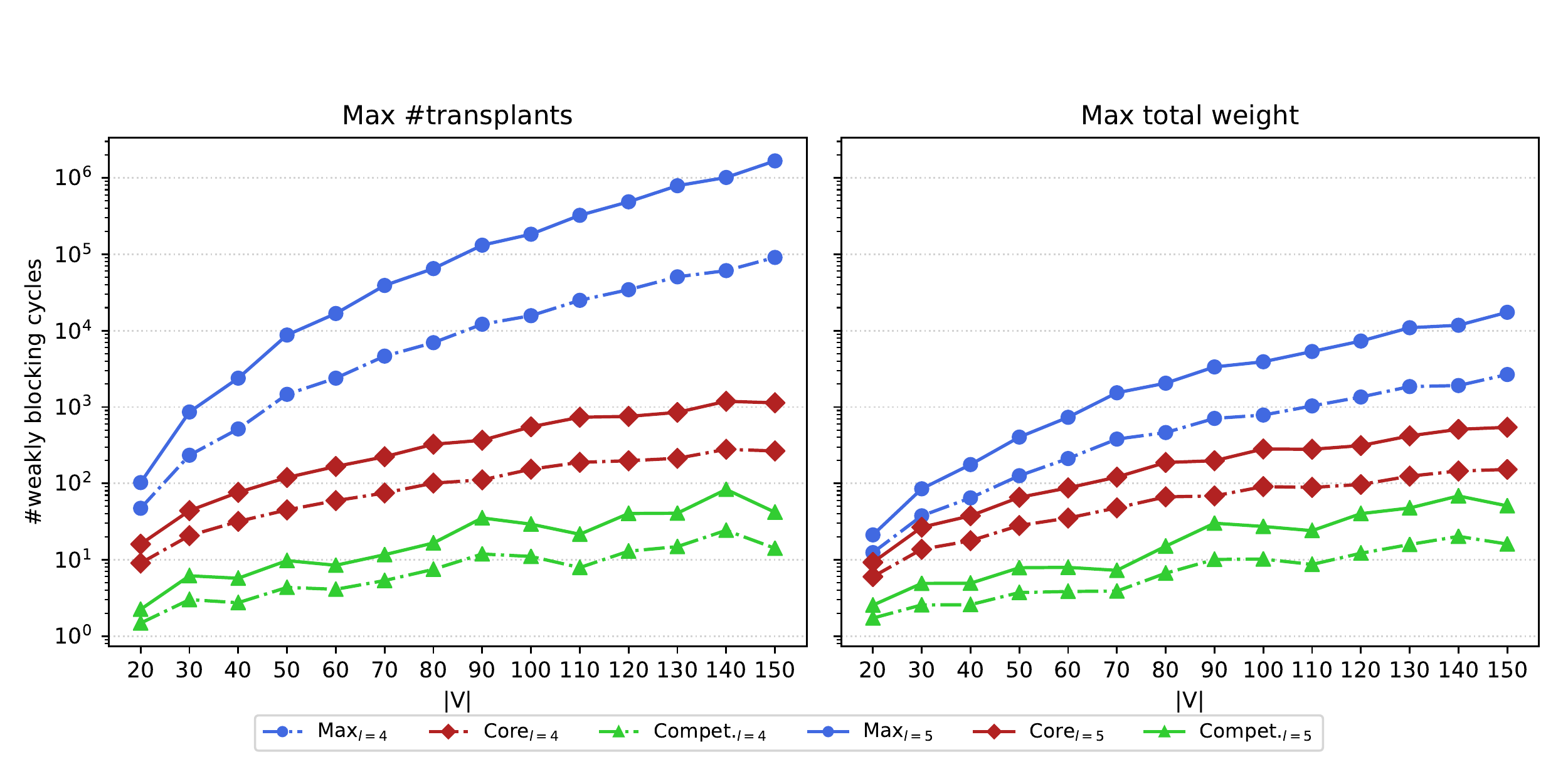}
    \caption{Number of weakly blocking cycles of size ${l=4}$ and $l=5$ for solutions with maximum number of transplants (left) and maximum total weight of transplants (right), for unbounded exchanges and weak preferences.}
    \label{fig:wb_up5}
\end{figure}
Although the information above is already insightful,  to complement our analysis  we provide in Figure \ref{fig:better_verts} information on the average number of vertices of an instance that strictly prefer their allotments in at least one weakly blocking cycle (i.e., on the number of patients that could receive a strictly better kidney in a deviating allocation). An important conclusion can be drawn from the results in the figure: the maximisation of total weight decreases the number of agents that can get a better allotment in a blocking cycle when compared to the maximum size solutions (compare curves {Max} in Figure \ref{fig:better_verts} (left) and (right)). It also allows, by comparison with Figure \ref{fig:pf_weak}, to analyse the trade-off that would be necessary to make in terms of reduction of the total number of transplants to meet a certain level of patients preferences.  



\begin{figure}[htbp]
    \centering
    \includegraphics[scale = 0.7]{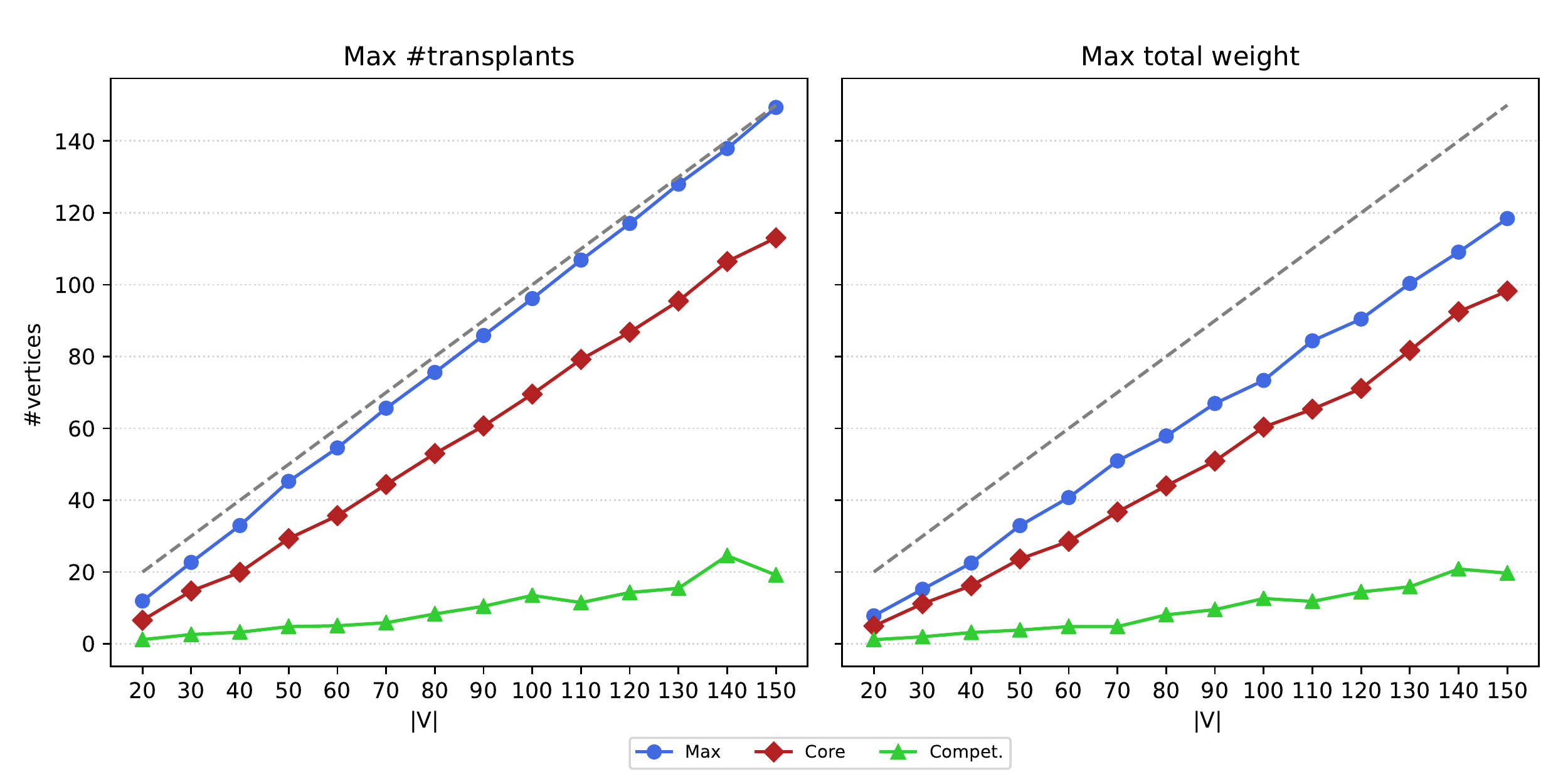}
    \caption{Average number of agents for an instance  for those there exists at least one weakly blocking cycles, where this agent receives a strictly better allotment for weak preferences. The grey line is a reference line showing the number of vertices in an instance.}
    \label{fig:better_verts}
\end{figure}

\subsection{CPU time for unbounded models}\label{sec:time}

In Table~\ref{tab:time} we present the average CPU time for solving an instance of a given size with one of the tree newly proposed IP models for unbounded case.
\begin{table}[htbp]
    \centering
    \footnotesize
    \begin{tabular}{c|rrr| rrr || rrr|rrr}

    &\multicolumn{3}{c|}{Max \# transplants} & \multicolumn{3}{c}{Max total weight} & \multicolumn{3}{c|}{Max \# transplants} & \multicolumn{3}{c}{Max total weight} \\
    $|V|$ & Core & Compet. & S.Core &Core & Compet. & S.Core & Core & Compet. & S.Core &Core & Compet. & S.Core \\
    \cline{1-13}
        &\multicolumn{6}{c|}{\bf Strict preferences}&\multicolumn{6}{c}{\bf Weak preferences}\\

20	&	0.00	&	0.03	&	0.01	&	0.00	&	0.02	&	0.01	&	0.00	&	0.04	&	0.01	&	0.00	&	0.03	&	0.01	\\
30	&	0.03	&	0.13	&	0.04	&	0.02	&	0.11	&	0.03	&	0.02	&	0.28	&	0.04	&	0.02	&	0.17	&	0.03	\\
40	&	0.08	&	0.48	&	0.12	&	0.06	&	0.25	&	0.11	&	0.09	&	0.63	&	0.10	&	0.06	&	0.44	&	0.08	\\
50	&	0.24	&	1.74	&	0.38	&	0.16	&	0.58	&	0.34	&	0.20	&	2.15	&	0.25	&	0.17	&	1.06	&	0.21	\\
60	&	0.47	&	2.39	&	0.87	&	0.28	&	0.91	&	0.79	&	0.52	&	6.03	&	0.44	&	0.26	&	2.87	&	0.39	\\
70	&	1.06	&	3.91	&	1.94	&	0.66	&	2.29	&	1.50	&	0.84	&	16.99	&	1.09	&	0.53	&	7.35	&	0.77	\\
80	&	1.62	&	6.54	&	3.26	&	0.82	&	3.39	&	2.32	&	1.41	&	32.21	&	1.63	&	0.76	&	17.47	&	1.01	\\
90	&	3.14	&	36.34	&	5.31	&	3.27	&	5.38	&	3.59	&	3.29	&	167.15	&	2.36	&	1.82	&	80.88	&	1.49	\\
100	&	3.53	&	16.19	&	19.26	&	2.43	&	6.15	&	9.81	&	4.51	&	188.35	&	8.87	&	3.08	&	95.39	&	4.62	\\
110	&	8.73	&	21.42	&	28.26	&	4.97	&	9.01	&	13.79	&	6.68	&	331.64	&	16.40	&	5.92	&	159.12	&	7.24	\\
120	&	17.84	&	72.87	&	57.36	&	6.81	&	15.36	&	24.32	&	20.14	&	392.88	&	19.60	&	6.79	&	218.58	&	10.87	\\
130	&	14.34	&	46.92	&	84.49	&	14.24	&	22.68	&	34.11	&	14.78	&	586.27	&	21.75	&	12.32	&	438.23	&	10.42	\\
140	&	29.50	&	61.99	&	110.82	&	21.51	&	34.33	&	46.67	&	41.59	&	708.92	&	40.97	&	16.43	&	539.56	&	14.89	\\
150	&	41.99	&	161.10	&	214.32	&	30.66	&	52.61	&	70.77	&	57.13	&	786.43	&	61.79	&	27.82	&	682.99	&	23.91	\\

    \end{tabular}
    \caption{Average CPU time (in seconds) for solving an instance of a given size with the proposed formulation.}
    \label{tab:time}
\end{table}

The instances with the weak preferences are more complicated, for core and, in particular, for the competitive allocation model. However, it was faster to find the strong core for weak, rather than for strict preferences. Moreover, surprisingly,  finding the strong core is the most time consuming task for strict preferences, while it is least time consuming for weak preferences. Finally, we can notice that models for finding core and strong core allocations are performing within the same ranges of magnitude with respect to the CPU time if compared with the corresponding models for the bounded case, analysed in~\cite{Klimentova2020stable}.

\subsection{Violation of respecting improvement property}\label{sec:violRI}




In this section we will make a computational analysis on how often the respecting improvement (RI) property is violated for different models, for both unbounded and bounded cases. To do so, for each model and for instances with 20 and 30 vertices we run the following procedure, presented in Algorithm~\ref{alg:RIcheck}. For the unbounded case  we considered the  Max and Core models under both objectives.

Let $r^i_k$, $r \in \{1,\dots, |V|\}$ be the rank of good $k$ for  agent $i$, that will reflect preferences of $i$,  i.e. if $r^i_k\leq(<, =) r^i_j$, then $kR_i(P_i, I_i)j$.

\begin{algorithm}[H]
\SetAlgoLined
\KwResult{$N$ number of violations of RI property}
$N \leftarrow 0$\;
\For{$i\in V$, $j\in V$, $i\neq j$}{
    Let $R$ be the current preferences of agents\;
    Find allocation with the best allotment for $i$ with respect to $R$, denote the solution by $\bar y$\;
    For $\bar y_{il} = 1$ denote $ \bar r \leftarrow r^i_l$\; 
    \While{$\exists kP_ji$}
     {
        Let $k$ be the first strictly preferred agent for $j$ that precedes $i$ in $R$\;
               
        \If {Strict preferences}{ Swap $i$ with $k$ in the list of preferences of $j$\;}
        \If {Weak preferences}{ Let $i$ become equally preferred for $j$ as $k$ (i.e. $r^j_i = r^j_k$)\;}
        Denote modified preferences by $\tilde R$\;
        Find allocation with the best allotment for $i$ with respect to $\tilde R$, denote solution by $\tilde y$;\\
       For $\tilde y_{it} = 1$, denote $\tilde r = r^i_t$\;    
      \If{$ \bar r< \tilde r$}{
             The respecting improvement property is violated:            $N\leftarrow N+1$\;
       }
       $\bar r \leftarrow \tilde r$; 
       $R \leftarrow \tilde R$\;
     }
}
 
 \caption{Procedure for Checking RI property }\label{alg:RIcheck}
\end{algorithm}
For each pair of agents $i$ and $j$, agent $i$ is consecutively making improvements, moving up in the preference list of agent $j$ until its top. In each step (see \textbf{while} loop in the algorithm) for the case of strict preferences, $i$ is swapped with $k$, who is the first strictly preferred agent by $j$ to $i$. For the case of ties, agent $i$ first becomes equally preferred for $j$ as $k$. After the improvements, the best allocations for the original ($R$) and improved  ($\tilde R$) preferences are compared for $i$. It is considered that there is a violation of the RI property if $i$ obtains a  strictly worse allotment in allocation for $\tilde R$.


Figures~\ref{fig:ri_strict} and~\ref{fig:ri_weak} present box plots for the number of violations of the RI property for instances of a given size 
for strict and weak preferences, respectively, for those models where the RI property is violated at least once. 
Models whose results are the same, independently of the objective considered, are plotted together. That is the case, for example, of Core$_{t}$ and Core$_{w}$, for $k=3$ and strict preferences (see figure~\ref{fig:ri_strict}), or of Wako-core and Core, weak preferences and $k = 2$ and $k = 3$ (see figure~\ref{fig:ri_weak}).

\begin{figure}
\includegraphics[scale = 0.7]{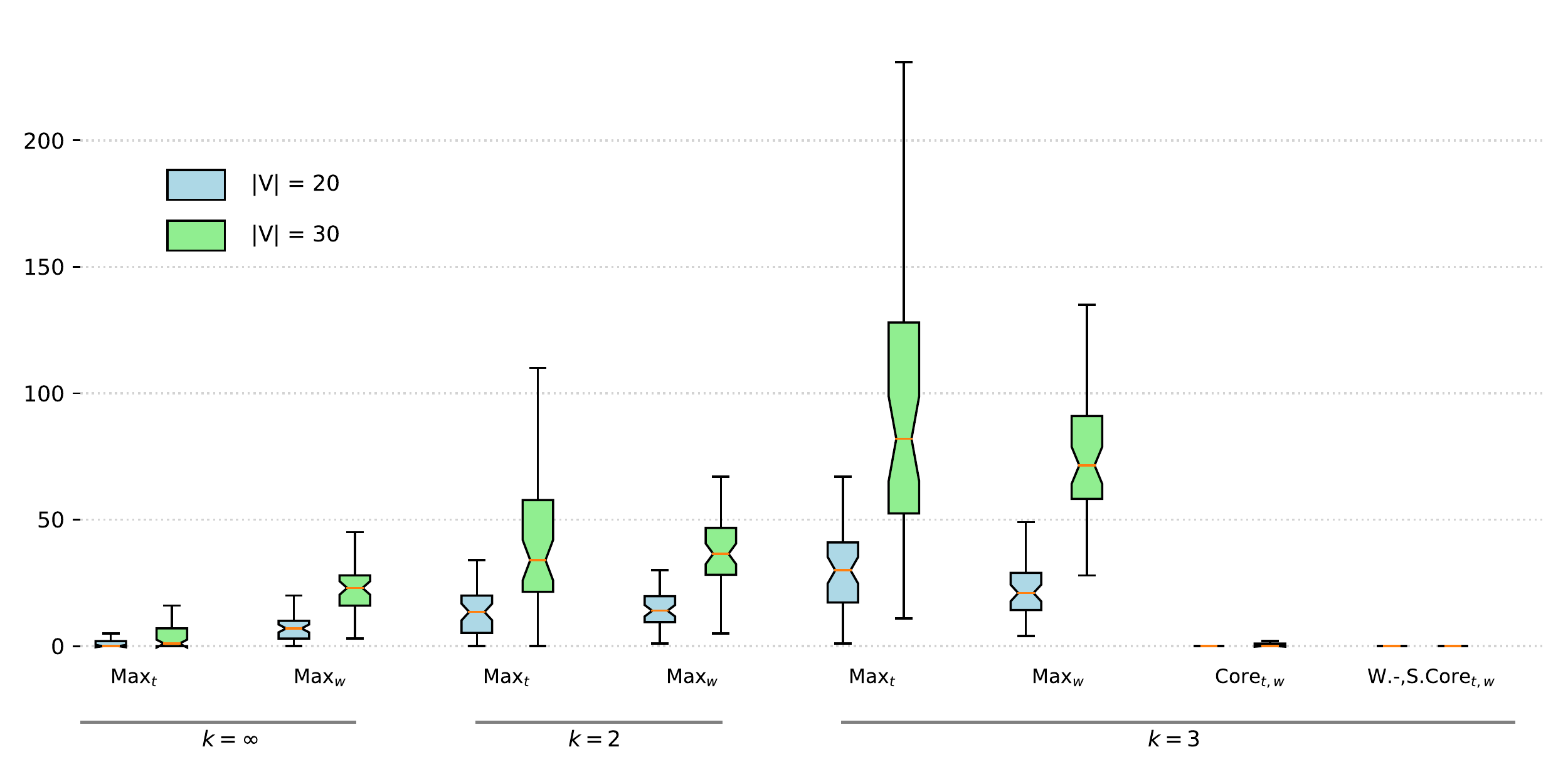}
    \caption{Number of violations of the respecting improvement property for all instances in total of a given size, $|V| = 20,30$, for strict preferences.}
    \label{fig:ri_strict}
\end{figure}
\begin{figure}
\includegraphics[scale = 0.7]{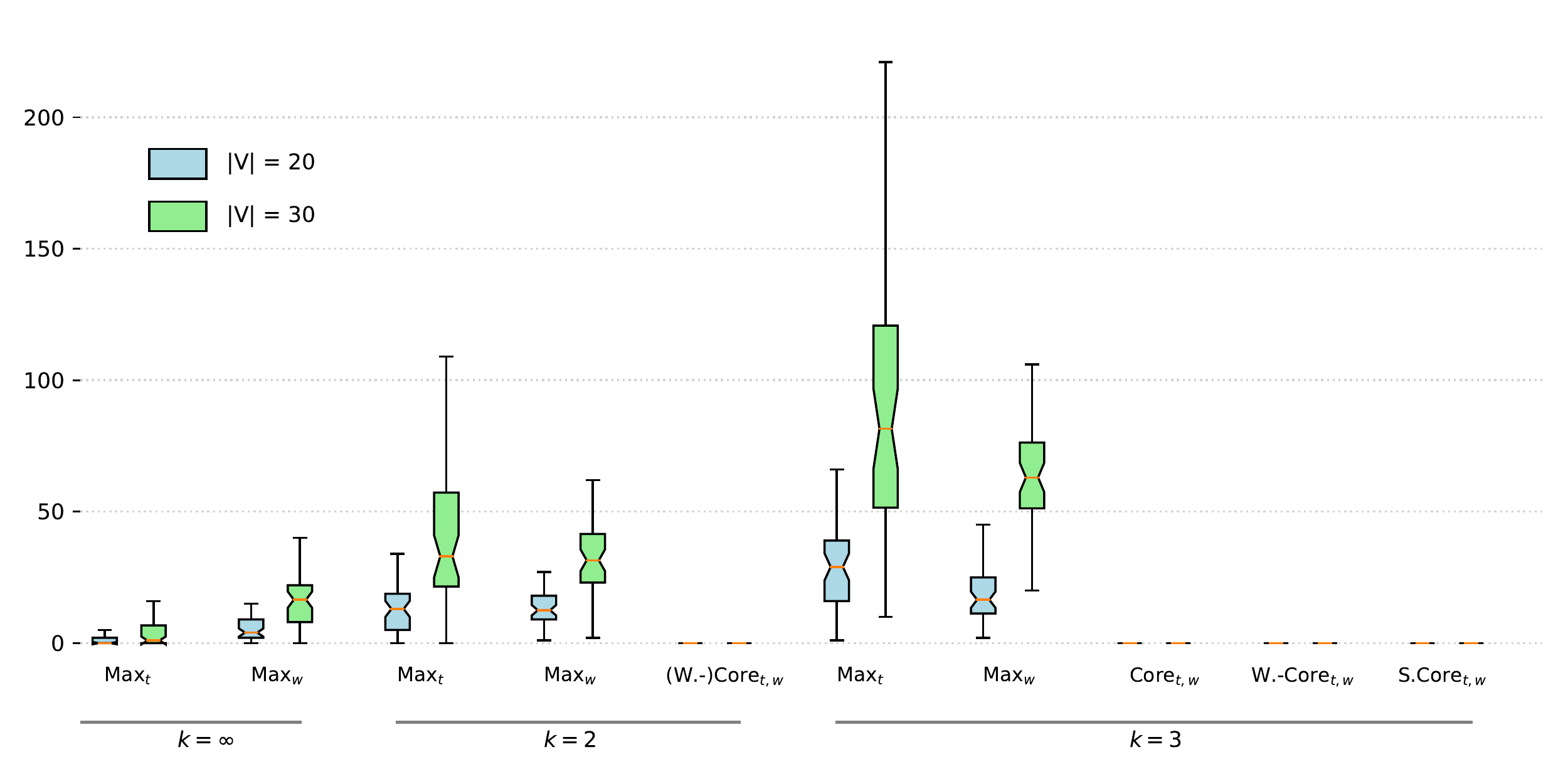}
    \caption{Number of violations of the respecting improvement property for all instances in total of a given size, $|V| = 20,30$, for weak preferences.}
    \label{fig:ri_weak}
\end{figure}

For (Wako-, Strong) Core models there were few cases of violations of the RI property, as reflected in the figure. To give an indication, the total number of violations for all instances with 30 vertices for the weak preferences and $k=3$ was 4549 for Max$_t$, 3145 for Max$_w$, 10 for Core$_{t,w}$, 20 for W.-Core$_{t,w}$, and 2 for S.Core$_{t,w}$. 
For maximum size and maximum weight solutions (Max$_t$ and Max$_w$, respectively), both for the unbounded and  the bounded cases, one can observe a significant number of violations.  Those numbers increase with  instance size. Interestingly, for the unbounded case the number of violations for Max$_t$  was lower than that for Max$_w$. This can be explained by the fact that the former problem has a larger number of alternative best allocations, while for the weighted objective problem the solution is usually unique. On the contrary, for the bounded case maximum weighted solutions violated the property less times, compared to maximum size solutions.

\section{Conclusion}

This paper advances current state of the art in several lines of research. We first prove that in case of strict preferences the unique competitive allocation respects improvement; an extension of that result is provided for the case of ties.

We also advance the work in the housing market of Shapley and Scarf presented in \cite{QW2004} by providing Integer Programming models that do not require exponential number of constraints for the weak core, strong core, and the set of competitive allocations. These models assume that there is no limit on the maximum size of an exchange cycle. However, since there are problems where such assumption may be difficult to hold (e.g. Kidney Exchange Programmes) we further propose alternative IP models for bounded cycles. This contribution is inspired by the definition of competitive equilibrium allocations provided in \cite{Wako1999}.

We proceed with computational experiments that provide insights on the trade-off between stability requirements and maximum number of transplants.  Results show that with increasing size of the instances, such trade-off decreases:  for instances with  more  than  50  nodes core allocations impact on the reduction of transplants is less  than  3\%, decreasing  to  1\%  for  the  largest instance.  

Furthermore, results show that when the objective is to maximise the number of transplants, the ``unstability'' of the solutions, measured by the number of weakly blocking cycles barely depends on the length of the exchanges. Additionally, the maximisation of total weight instead of the number of transplants, leads to solutions where patients' preferences matter more.

As the main open question we left open whether the respecting improvement property with regard to the best allotment holds for a) the core for unbounded exchanges b) for stable matchings in the roommates problem with strict preferences. It would also be interesting to study whether the respecting improvement property can be used to characterise the TTC mechanism for the classical housing markets with strict preferences. 

\section*{Acknowledgements}

We thank Antonio Nicol\'o for his contribution to an earlier version of this paper, and
Tayfun S\"onmez, and Utku \"Unver for valuable comments. 

\bibliographystyle{plain}
\bibliography{KEPstable-v1}
\end{document}